\documentclass[letterpaper, 10 pt, conference]{ieeeconf}  % Comment this line out if you need a4paper

\IEEEoverridecommandlockouts                              % This command is only needed if 
\usepackage{bm}
\newtheorem{cconj}{Conjecture}

\newtheorem{pproblem}{Problem}
\newenvironment{problem}{\begin{pproblem}\rm }{%\hfill \hspace*{1pt} \hfill $\lrcorner$
\end{pproblem}}
\newtheorem{ddefn}{Definition}

\newtheorem{llemma}{Lemma}
\newenvironment{lemma}{\begin{llemma}\rm }{%\hfill \hspace*{1pt} \hfill $\lrcorner$
\end{llemma}}
\newtheorem{ttheorem}{Theorem}

\newtheorem{aassumption}{Assumption}
\newenvironment{assumption}{\begin{aassumption}\rm }{%\hfill \hspace*{1pt} \hfill $\lrcorner$
\end{aassumption}}
\newtheorem{cclaim}{Claim}

\newtheorem{pproposition}{Proposition}
\newenvironment{proposition}{\begin{pproposition}\rm }{%\hfill \hspace*{1pt} \hfill $\lrcorner$
\end{pproposition}}
\newtheorem{rremark}{Remark}
\newenvironment{remark}{\begin{rremark}\rm }{%\hfill \hspace*{1pt} \hfill $\lrcorner$
\end{rremark}}

\newtheorem{eexample}{Example}

\newtheorem{ccorollary}{Corollary}
\newenvironment{corollary}{\begin{ccorollary}\rm }{\hfill \hspace*{1pt} \hfill %$\lrcorner$
\end{ccorollary}}

\usepackage{graphics} % for pdf, bitmapped graphics files
\usepackage{amsmath} % assumes amsmath package installed
\usepackage{amssymb}  % assumes amsmath package installed
\usepackage{epstopdf}
\usepackage{cite}

\usepackage{todo}

\usepackage{comment}
\usepackage[normalem]{ulem}

\usepackage{overpic}
\usepackage[dvipsnames]{xcolor}
\usepackage[normalem]{ulem}
\newcommand\redsout{\bgroup\markoverwith{\textcolor{red}{\rule[0.5ex]{2pt}{1.0pt}}}\ULon}

\DeclareMathOperator{\argmin}{\operatorname{argmin}}

\newcommand\N{\mathbb{N}}

\newcommand\R{\mathbb{R}}

\newcommand\sign{\text{sign}}

\usepackage{color}

\begin{document}

\title{\LARGE \bf
Semi-Explicit Solutions to the Prying-Pedestrian Surveillance-Evasion Differential Game and Extensions to Two Pursuers
}

\author{{Philipp Braun}$^{1}$, Daniel Coutinho$^{2}$, Timothy L. Molloy$^{3}$ and Iman Shames$^{4}$ 
 \thanks{$^{1}$P. Braun, is with the School of Engineering, Australian National University, Canberra, Australia, {\tt\small philipp.braun@anu.edu.au}.}
 \thanks{$^2$D. Coutinho is with the Department of Automation and Systems Engineering, Universidade Federal de Santa Catarina, Florianópolis,  Brazil
 {\tt\small daniel.coutinho@ufsc.br}.}
 \thanks{$^{3}$T. L. Molloy is with the Department of Electrical and Computer Systems Engineering, Monash University, Clayton, Australia, {\tt\small  timothy.molloy@monash.edu}.}
 \thanks{$^{4}$I. Shames is with the 
Department of Electrical and Electronic Engineering, University of Melbourne, Melbourne, Australia, {\tt\small iman.shames@unimelb.edu.au}.}%
\thanks{This work was supported by the United States Air Force Office of Scientific Research under Grant No. FA2386-24-1-4014.}
 }%

\maketitle

\begin{abstract}
In \cite{braun_prying_2025}, the authors recently proposed and solved a surveillance-evasion differential game in which an agile pursuer (the prying pedestrian) seeks to remain within a given surveillance range of a less agile evader for as long as possible while the evader seeks to escape as quickly as possible.
In this paper, we provide initial results that extend this game from the 1 versus 1 (1v1) setting to a 2 versus 1 (2v1) setting with two pursuers and one evader. 
By deriving and exploiting semi-explicit or geometric reinterpretations of the existing 1v1 results, we derive partial solutions to the 2v1 game for the case of static pursuers and for the case of an evader that is at least twice as fast as the pursuers. 
While the 2v1 results of this paper build on the 1v1 results of \cite{braun_prying_2025}, a different solution approach is developed to avoid a coordinate transformation that reduces the 1v1 game to two dimensions but which is ineffective at simplifying the 2v1 game.
Beyond enabling progress on the 2v1 game, our new approach yields new geometric interpretations of the optimal pursuer and evader strategies in the 1v1 game, and opens further possible extensions.
\end{abstract}

\IEEEpeerreviewmaketitle

\section{Introduction}
\label{sec:intro}

Differential games of pursuit evasion \cite{merz_homicidal_1974,exarchos_asymmetric_2014,exarchos_suicidal_2015,weintraub_introduction_2020,dorothy_one_2024}, surveillance evasion \cite{lewin_surveillance-evasion_1975,lewin_conic_1979,lewin_isotropic_1989,braun_prying_2025}, and collision avoidance \cite{merz_optimal_1973,molloy_optimal_2020} have been extensively studied in 1 versus 1 (1v1) settings with two opposing players (or agents) since the work of Isaacs \cite{isaacs_differential_1965}.
Pursuit-evasion and collision-avoidance games have also received attention in settings with multiple pursuers and multiple evaders \cite{ibragimov_simple_2018,weintraub_introduction_2020,azamov_evasion_2025,garcia_multiple_2021,exarchos_uav_2016}, with a notable concentration of work on 2 versus 1 (2v1) settings with two pursuers and 1 evader beginning in \cite{hagedorn_differential_1976}.
In contrast, there appears to have been limited (if any) detailed consideration of surveillance-evasion differential games with more than one pursuer and evader.
We therefore seek to extend the recent prying-pedestrian surveillance-evasion differential game of \cite{braun_prying_2025} from a 1v1 setting of a single pursuer surveilling a single evader to a 2v1 setting involving two pursuers surveilling an evader.

Lewin and Breakwell \cite{lewin_surveillance-evasion_1975} were the first to pose and solve a \emph{game-of-degree} surveillance-evasion game in which a turn-rate limited pursuer seeks to maximize the time within surveillance range of an evader seeking to escape in minimum time, with the evader being agile in the sense of being capable of instantaneously changing its heading.
Several varieties of 1v1 surveillance-evasion games of degree have since been posed and solved \cite{lewin_conic_1979,lewin_isotropic_1989}, with the most recent being the prying-pedestrian game of \cite{braun_prying_2025} in which an agile pursuer (the prying pedestrian) seeks to maximize the time it surveils a faster turn-rate limited evader seeking to escape in minimum time.

The prying-pedestrian game of \cite{braun_prying_2025} is of practical and theoretical interest for two key reasons.
Firstly, it offers practical insight into the performance attainable by any turn-rate limited pursuer since an agile pursuer can employ the surveillance strategies of less agile pursuers.
Secondly, its solution exhibits surprisingly complex speed-ratio dependence, which does not appear in similar pursuit-evasion or collision-avoidance games with agile pursuers (cf.~\cite{exarchos_suicidal_2015} and \cite{molloy_optimal_2020}).
Indeed, when the evader is twice as fast as the pursuer, the evader's optimal strategies exhibit counter-intuitive phenomena, including turning until aligned with the pursuer and proceeding to ``pass through it''.
This speed-ratio dependence was, however, only observed in \cite{braun_prying_2025} via the numerical solution of differential equations, and an explicit or geometric (rather than numerical) justification for its optimality is currently lacking.
In this paper, we seek an explicit or geometric (and intuitive) reinterpretation of the evader's strategies to progress the solution of the 2v1 prying-pedestrian differential game of two agile pursuers surveilling a single evader. 

The main contributions of this paper are the development of (semi-)explicit solutions to the 1v1 prying pedestrian game of \cite{braun_prying_2025} (including the first geometric characterization of the evader's optimal strategies at different speed ratios); and, the formulation and preliminary investigation of a 2v1 version of the prying-pedestrian differential game with two agile pursuers and one (faster) turn-rate limited evader.
We also pay particular attention to analyzing the case of two static (or stationary) pursuers since it can be equivalently viewed as a novel optimal control problem that extends the minimum-time circle-escape problem of \cite{molloy_minimum-time_2023,weintraub_minimum_2024} to the minimum-time escape from two (potentially overlapping) circles.

This paper is structured as follows.
In Section \ref{sec:setting}, we introduce the 2v1 game prying-pedestrian differential game and revisit the 1v1 game.
In Section \ref{sec:1v1_setting}, we reinterpret, and develop geometric arguments for, optimal strategies in the 1v1 game.
In Sections \ref{sec:2v1} and \ref{sec:non-static-pursuers}, we present preliminary characterizations of solutions to the 2v1 game.
Conclusions and future work are provided in Section \ref{sec:conclusions}.

{\bf Notation:} 
A closed
ball of radius $\rho>0$ centered around $\xi\in \R^2$ is denoted by $\mathcal{B}_\rho(\xi)=\{\zeta \in \R^2| \ |\xi-\zeta|\leq \rho\}$.
For $\theta \in \R$ the rotation matrix is 
$
    R(\theta) = \left[ \begin{smallmatrix}
         \cos(\theta) & -\sin(\theta)  \\
         \sin(\theta) & \cos(\theta) 
    \end{smallmatrix} \right]
$.
The sign function is the set-valued map 
\begin{align}
\sign(x) = \left\{ \begin{array}{cl}
    \{-1\} & \text{if } x<0,  \\
    \{1\} & \text{if } x>0, \\
    \{-1,1\} & \text{if } x=0.
\end{array} \right. \label{eq:sign_f}    
\end{align}
For $[\begin{smallmatrix}
     x \\ y
 \end{smallmatrix}] \in \R^2 \backslash\{0\}^2$, we write
$
 \arctan_2 ([\begin{smallmatrix}
     x \\ y
 \end{smallmatrix}])= \arctan_2 ( y,x).
$

\section{Setting and Problem Formulation} 
\label{sec:setting}

In this paper, we extend the results derived in \cite{braun_prying_2025} towards settings with multiple pursuers. 
In this section, we extend the 1v1 prying-pedestrian surveillance-evasion differential game of degree of \cite{braun_prying_2025} to the 2v1 setting of two pursuers versus one evader. 
Optimal and suboptimal solutions for different parameter settings are derived in the following sections.

To introduce the 2v1 game, we consider an evader and two pursuers moving in the two-dimensional Euclidean plane.
The evader maintains a constant speed $v_e > 0$ and its position $\xi_e = [x_e,y_e]^\top \in \mathbb{R}^2$ and heading angle $\theta_e \in (-\pi,\pi]$ evolve according to the unicycle (or Dubins car) kinematic equations
\begin{align}
   \label{eq:cartesianDynamics}
\dot{\xi}_{e\theta} =   \begin{aligned}
   \left[\begin{array}{c}
        \dot{x}_e  \\
        \dot{y}_e \\
        \dot{\theta}_e
   \end{array} \right] = f_e(\xi_{e\theta},u_e) = \left[\begin{array}{c}
       v_e \sin \theta_e(t)  \\
       v_e \cos \theta_e(t)  \\
       \omega_e u_e(t)
   \end{array}\right].
    \end{aligned}
\end{align}
Here, $\omega_e > 0$ is the evader's finite maximum turn rate in radians per second.
The evader's control input is its normalized turn rate 
$u_e(t) \in [-1,1]$ for all $t\in \R_{\geq 0}$ with $u_e(t)=1$, $u_e(t) =-1$ and $u_e(t)=0$ corresponding to a right-hand turn, a left-hand turn, and going straight, respectively.
We use the notation $\xi_e=[x_e,y_e]^\top$ to capture the position of the evader and $\xi_{e\theta}=[x_e,y_e,\theta]^\top$ to capture the position and the orientation.

The two pursuers $i\in \{1,2\}$ similarly move with constant speeds $v_{p_i} \geq 0$ but their positions $\xi_{p_i} = [x_{p_i},y_{p_i}]^\top \in \mathbb{R}^2$ and headings $\theta_{p_i} \in (-\pi,\pi]$ are described by the (agile) kinematic equations
\begin{align}
   \label{eq:cartesianDynamicsPursuer}
   \begin{aligned}
\dot{\xi}_{p_i} = \left[ \begin{array}{c} 
     \dot{x}_{p_i}  \\
      \dot{y}_{p_i} 
\end{array}\right]  = 
f_{p_i} (\xi_{p_i},u_{p_i}) 
= \left[ \begin{array}{c}
     v_{p_i} \sin \theta_{p_i}(t)  \\
     v_{p_i} \cos \theta_{p_i}(t) 
\end{array} \right]. 
   \end{aligned}
\end{align}
The pursuers are agile in the sense that they are capable of instantaneously changing their heading angles $\theta_{p_i}(t)$ (i.e. they have an infinite maximum turn rate).
Thus, $u_{p_i}(t) = \theta_{p_i}(t) \in (-\pi, \pi]$, $t\in \R_{\geq 0}$, are the pursuers' control inputs.

To pose the 2v1 prying pedestrian surveillance-evasion game, 
we define the range between
the evader and each pursuer $i\in \{1,2\}$ as 
\begin{align}
    r_i
    & = |\xi_{p_i} - \xi_e| = \sqrt{(x_{p_i} - x_e)^2 + (y_{p_i} - y_e)^2}. \label{eq:def_r}
\end{align}
The evader is in a pursuer's surveillance range $\rho>0$ if $r_i\leq \rho$ for $i\in \{1,2\}$.

The objective of the evader is to ``escape'' from both pursuers by increasing 
the range $r_i(t)$ beyond the surveillance range $r_i(t)> \rho$, i.e., to achieve $r_i(T) > \rho$ for all $i\in\{1,2\}$ for some $T\geq 0$ and to achieve this property in minimal time.  
Conversely, the objective of the pursuers is to keep at least one pursuer within surveillance range of the evader (i.e. to keep $r_i(t) \leq \rho$ for all $t \geq 0$ for some $i \in \{1,2\}$). 
For the analysis, we focus on the game of degree (instead of the game of kind) and maximize the time the pursuers can keep the evader under surveillance. The game of kind (whether the evader can escape or the pursuers can keep it under surveillance indefinitely) can be solved using similar ideas as in  \cite{braun_prying_2025} independent of the number of pursuers.
The 2v1 game of degree is summarized as follows.

\begin{problem}[Game of Degree] \label{prob:game_of_degree}
Consider the dynamics \eqref{eq:cartesianDynamics} of the evader and the dynamics \eqref{eq:cartesianDynamicsPursuer} of $i\in\{1,2\}$ pursuers with inputs $u_e\in [-1,1]$ and $\theta_{p_i}\in (-\pi,\pi]$, $i\in \{1,2\}$, and defined through parameters $\omega_e,v_e\in \R_{>0}$, $v_{p_i}\in \R_{\geq 0}$ with $v_{p_i}<v_e$ for all $i\in \{1,2\}$.
Additionally, define $u_p=[u_{p_1},u_{p_2}]^\top$.
The game of degree with surveillance radius $\rho>0$, is the optimization problem
\begin{align}
V(\xi_{e\theta 0},& \{\xi_{{p_i}0}\}_{i=1}^2) = \min_{u_e} \max_{u_p}  \int_0^T 1 \, \mathrm{d}t  \label{eq:problem_constraints} \\
\begin{split}
\text{s. t. } \quad & \dot{\xi}_{e\theta}(t) = f_{e}(\xi_{e\theta}(t), u_e(t)), \quad \xi_{e\theta}(0) = \xi_{e\theta0}, \\
& \dot{\xi}_{p_i}(t) = f_{p_i}(\xi_{p_i}(t),  u_{p_i}(t)), \quad \xi_{p_i}(0) = \xi_{{p_i}0}, \\
& u_e(t) \in [-1,1], \quad t \in [0,T],\\
& u_p(t) \in \mathbb{R}^2, \quad t \in [0,T],\\
& \inf_{\substack{|\xi_e(T)-\xi_{p_i}(T)| > \rho \\ i\in \{1,2\}}} T. \nonumber
\end{split} 
\end{align}
\end{problem}

\begin{remark}
    The condition $v_{p_i}<v_e$, $i\in \{1,2\}$, ensures that Problem \ref{prob:game_of_degree} is well-posed, i.e., the game ends in finite time. It is straightforward to extend the problem formulation to $n\in \N$ pursuers with individual surveillance radii $\rho_i\in \R_{>0}$. Here, and in the rest of the paper, we focus on $n=1$ and $n=2$ and on $\rho_1=\rho_2=\rho$ for simplicity of notation.
\end{remark}

In the following, by leveraging the results of \cite{braun_prying_2025} for the 1v1 setting, under additional assumptions on the speed ratio, we derive initial results for the 2v1 setting. 
However, in the 1v1 setting of \cite{braun_prying_2025}, an evader centric coordinate system is used to reduce the dimensionality of the problem to derive optimal strategies.
This approach is not as effective in the 2v1 setting, and thus here we use a different coordinate system defined through the initial positions of the evader and the pursuers.
We therefore begin by recalling and reinterpreting the 1v1 results derived in \cite{braun_prying_2025} in this different coordinate system.

\section{The 1v1 prying pedestrian pursuit evasion game of degree} 
\label{sec:1v1_setting}

In this section, we reinterpret the results of \cite{braun_prying_2025} to derive a (semi-)explicit solution to the 1v1 version of Problem \ref{prob:game_of_degree}.

\subsection{Intuitive optimal strategies for the 1v1 setting}

\begin{proposition} \label{prop:main_result_journal1v1}
Consider the game of degree in Problem \ref{prob:game_of_degree} reduced to a single pursuer. 
Let $\rho, v_e,v_{p_{1}},\omega_e\in \R_{>0}$, $v_e > v_{p_{1}}$
and let $\xi_{e\theta 0}\in \R^3$ and $\xi_{{p_1} 0}\in \R^2$ denote arbitrary initial conditions with $|\xi_{p_1}-\xi_e|\leq \rho$. 
Then,
(a) the optimal solution $T^*=V(\xi_{e\theta 0},\xi_{{p_1}0})$ is finite;
(b) the optimal strategy of the evader satisfies
\begin{align}
    u_e(t)\in \{-1,0,1\} \qquad \forall t\in [0,T^*]
\end{align}
with the evader changing its input at most once from $|u_e(t)|=1$ to $u_e(t)=0$; and,
(c) the optimal strategy of the pursuer is piecewise constant and the pursuer changes its input at most once. 
If additionally the condition $v_e\geq 2 v_{p_1}$ is satisfied, then the pursuer's strategy is constant.
\end{proposition}
\begin{proof}
    Follows from \cite[Thms. 1-3, Rem. 10]{braun_prying_2025}.
\end{proof}

Instead of an evader-centric coordinate system used in \cite{braun_prying_2025}, we now use a coordinate system defined through the initial conditions of the evader and the pursuer in this section to illustrate Proposition \ref{prop:main_result_journal1v1}. Moreover, by using Proposition \ref{prop:main_result_journal1v1} we characterize a solution of the 1v1 setting under the following additional assumption on the speed ratio of the evader and the pursuer, simplifying the strategy of the pursuer.

\begin{assumption}\label{as:speed_ratio}
    For all $i\in \{1,2\}$, the speed ratio of the evader and pursuer $i$ satisfies $v_e \geq 2 v_{p_i}$.
\end{assumption}

An illustration highlighting the importance of the assumption is given later in the paper in Lemma \ref{lem:illustration_speed_ratio_assumption}.
To reinterpret solutions $u_e^*(\cdot;\xi_{e\theta 0}, \xi_{{p_1}0})$, $u_{p_1}^*(\cdot;\xi_{e\theta 0}, \xi_{{p_1}0})$, $V(\xi_{e\theta 0}, \xi_{{p_1}0})$ to the 1v1 game, we consider a coordinate system defined through the initial conditions such that at time $t=0$,
\begin{align}
    \xi_e^E(0)=0=\xi_{e0}^E \quad  \text{and} \quad \theta_e^E(0)=\tfrac{\pi}{2}
\end{align}
with the superscript $\cdot^E$ denoting variables in this coordinate system.
The setting is visualized in Fig.~\ref{fig:1v1_setting}.
\begin{figure}[htb]
    \centering
    \begin{overpic}[width = 1\columnwidth]{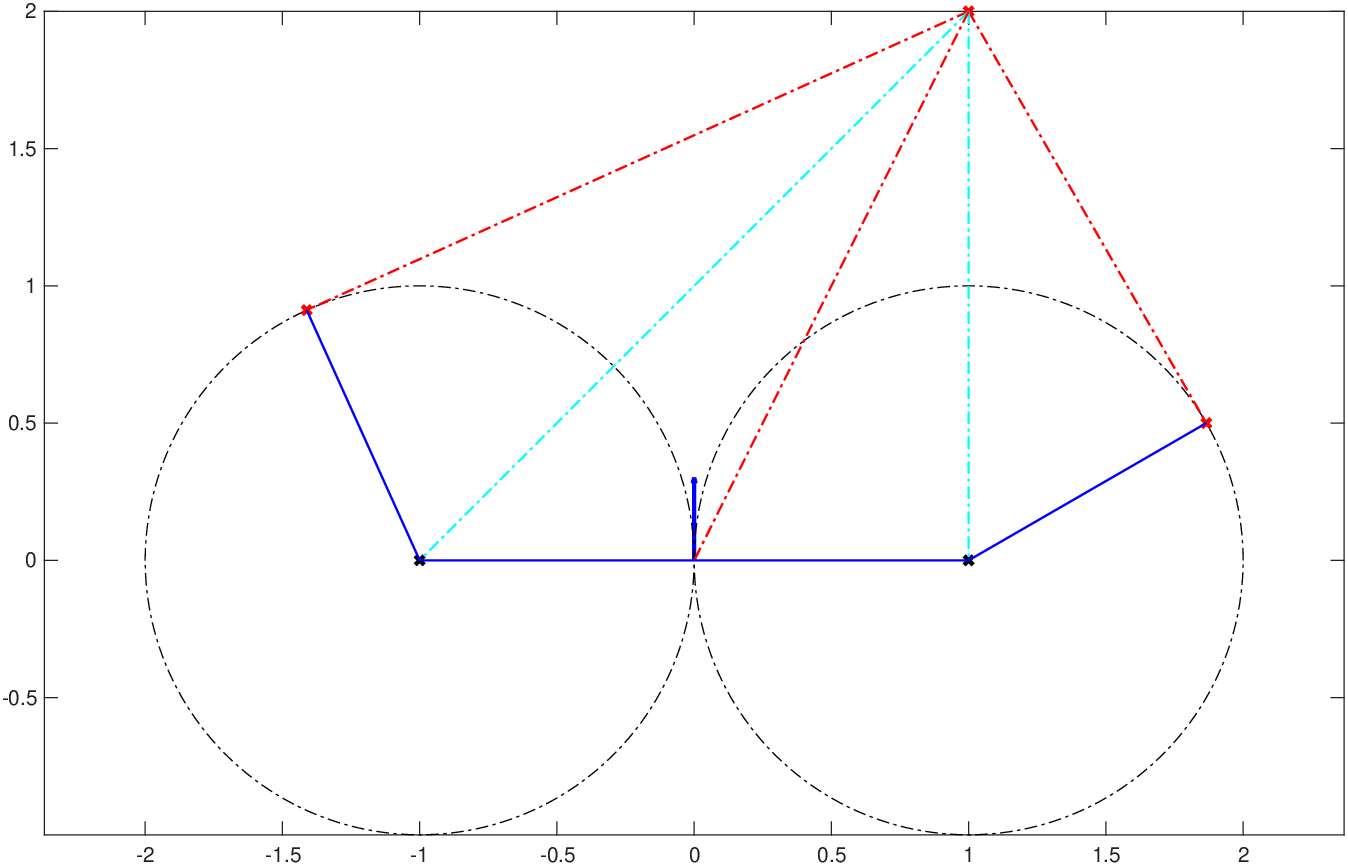}
    \put(53,18){\small{$\xi_{e0}^E$}}
    \put(73,18){\small{$C_1^E$}}
    \put(25,18){\small{$C_{-1}^E$}}
    \put(17,44){\small{$S_{-1}^E$}}
    \put(74,60){\small{$\xi_{p0}^E$}}
    \put(52,51){\small{$\alpha_{-1}$}}
    \put(32,25){\small{$\beta_{-1}$}}
    \put(80,51){\small{$\alpha_{1}$}}
    \put(70,25){\small{$\beta_{1}$}}
    \put(92,30){\small{$S_{1}^E$}}
    \end{overpic}
    \caption{Coordinate system defined through the initial state of the evader. Under the assumption that $u_e^*\in \{-1,0,1\}$ the evader turns around $C_q$, $q\in \{-1,1\}$ or follows a straight line.}
    \label{fig:1v1_setting}
\end{figure}
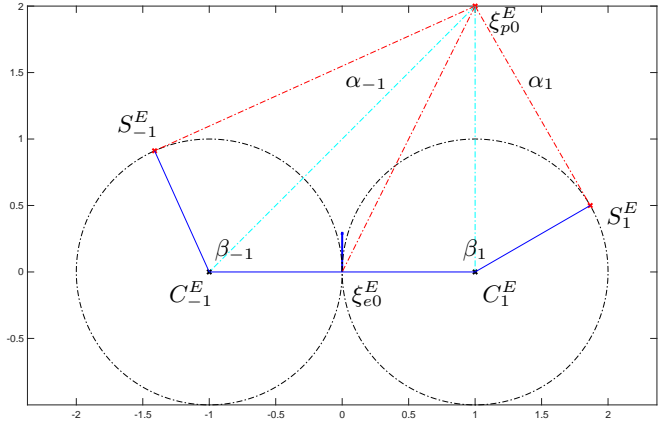
Assuming that the evader uses a constant input $u_e^*\in \{-1,0,1\}$ (which is a valid assumption according to Proposition \ref{prop:main_result_journal1v1}), the evader stays on one of two circles with radius 
\begin{align}
r_e= \frac{v_e}{\omega_e} \label{eq:turning_radius}
\end{align}
centered at 
\begin{align}
 C_{q}^E= \left[ \begin{array}{c} q r_e \\ 0 \end{array} \right], \qquad q\in\{-1,1\},   \label{eq:C_q}
\end{align}
or the evader stays on the $y$-axis following a straight line. More precisely, $u_e=q$ corresponds to a turn around $C_q^E$ and $u_e=0$ corresponds to a straight line.

The optimal strategy of the evader depends on the bearing angle. If the pursuer is to the right (i.e., $x_{p_10}^E\geq 0$) then the evader turns left (around $C_{-1}^E$) if the pursuer is to the left (i.e., $x_{p_10}^E\leq 0$) then the evader turns right (around $C_{1}^E$). The evader stops turning when the pursuer is aligned with the evader and centered behind the evader (with respect to the orientation $\theta_e^E$). The point where the evader stops turning is highlighted through $S^E_q$ in Fig.~\ref{fig:1v1_setting}. 
The pursuer's optimal strategy is to anticipate this evader maneuver by following a straight line defined through the points $\xi_{p_10}^E$ and $S^E_q$.

\subsection{Geometric derivation of solutions for the 1v1 setting} \label{sec:1v1_geometric_solution}

We now derive the parameters shown in Fig.~\ref{fig:1v1_setting} characterizing optimal solutions of the evader and the pursuer.
As a first step, note that for $q\in \{-1,1\}$, from $r_e=|S_q^E-C_{q}^E|$, if $|C_{q}^E-\xi_{p0}^E|\geq r_e$, it follows that 
\begin{align}
    \alpha_q = \arcsin\left( \frac{r_e}{|C_{q}^E-\xi_{p0}^E|} \right).
\end{align}
Thus, the distance between $S_q^E$ and $\xi_{p0}^E$ is given by
\begin{align*}
    |S_q^E-\xi_{p0}^E| = \sqrt{|C_{q}^E-\xi_{p0}^E|^2 - r_e^2} 
\end{align*}
and the point $S_q^E$ can be obtained from
\begin{align*}
    \frac{S_q^E-\xi_{p0}^E}{|S_q^E-\xi_{p0}^E|}=  R(\alpha_q)^q \frac{C_{q}^E-\xi_{p0}^E}{|C_{q}^E-\xi_{p0}^E|}
\end{align*}
i.e., 
\begin{align}
    S_q^E &= \frac{|S_q^E-\xi_{p0}^E|}{|C_{q}^E-\xi_{p0}^E|} R(\alpha_q)^q (C_{q}^E-\xi_{p0}^E) +\xi_{p0}^E. \label{eq:def_S} 
\end{align}
Here $R(\alpha_1)^1$ corresponds to an anti-clockwise rotation and $R(\alpha_{-1})^{-1}=R(\alpha_{-1})^{\top}$ corresponds to a clockwise rotation, depending on the centers $C_{1}^E$ and $C_{-1}^E$, respectively.
The angle $\beta_q$ characterizing the arc from $\xi_{e0}^E$ to $S_q^E$ is given by
\begin{align}
    \beta_q &= \left\{ \! \!\! \! \begin{array}{cc}
         \arccos \left( \tfrac{r_e^2- (S_q^E)^\top C_{q}^E}{r_e^2}    \right) & \text{if } S_q^E(2) \geq 0,  \\
         2\pi \! - \! \arccos \left( \! \tfrac{r_e^2- (S_q^E)^\top C_{q}^E}{r_e^2}    \right) &  \text{if } S_q^E(2) \leq 0,
    \end{array} \right. \label{eq:beta}
\end{align}
where $S_q^E(2)$ denotes the second component of $S_q^E=[S_q^E(1),S_q^E(2)]^\top$ and where
\begin{align*}
    \left[ \begin{smallmatrix} 0 \\ 1 \end{smallmatrix} \right]^\top \! (S_q^E -C_{q}^E) = \left[ \begin{smallmatrix}
     0 \\ 1 \end{smallmatrix} \right]^\top \! \left(S_q^E -\left[ \begin{smallmatrix} q r_e \\ 0 \end{smallmatrix} \right] \right) =S_q^E(2)
\end{align*}
decides if the angle $\beta_q$ is larger or smaller than $\pi$.

Under the assumption that the evader follows the circle of radius $r_e$ centered at $C_{q}^E$, the evader's position satisfies
\begin{align}
    \xi_{e}^E(t) = -q r_e R(t \tfrac{v_e}{r_e})^{-q} \left[ \begin{smallmatrix}
         1  \\
         0 
    \end{smallmatrix}\right] + C_{q}^E
\end{align}
and the evader reaches the position $S_q^E$  at time 
\begin{align}
    \bar{t}_e = \tfrac{\beta_q}{v_e} r_e. 
    \label{eq:def_t_e}
\end{align}
The pursuer travels along the line segment between $\xi_{p0}^E$ and $S_q^E$ satisfying
\begin{align}
 \xi_{p}^E(t) = \xi_{p0}^E+t v_p \frac{S_q^E-\xi_{p0}^E}{|S_q^E-\xi_{p0}^E|}.
\end{align}
and
reaches the point $S_q^E$ at time $\bar{t}_p=\frac{|S_q^E-\xi_{p0}^E|}{v_p}$.

Before we proceed, we show that under Assumption \ref{as:speed_ratio} and $q\in -\sign(x_{p0}^E)$\footnote{In the case that $x_{p0}^E=0$, the evader can turn left or right, explaining the definition of the sign function in \eqref{eq:sign_f}.} it holds that $\bar{t}_p\geq \bar{t}_e$ and thus the evader reaches the point $S_q^E$ before the pursuer. This is not necessarily the case if $v_e < 2 v_{p}$ and explains \cite[Fig. 9, right]{braun_prying_2025}.

\begin{lemma} \label{lem:illustration_speed_ratio_assumption}
    Consider the dynamics \eqref{eq:cartesianDynamics} and \eqref{eq:cartesianDynamicsPursuer}, let $\xi_{e0}^E=[0,0,\tfrac{\pi}{2}]^\top$ and $\xi_{p0}^E\in \R^2$ and consider the definitions in \eqref{eq:C_q}-\eqref{eq:beta} with $q\in-\sign(x_{p0}^E)$. Then, under Assumption \ref{as:speed_ratio} it holds that $\bar{t}_p=\frac{|S_q^E-\xi_{p0}^E|}{v_p}\geq \bar{t}_e = \frac{\beta_{q}}{v_e}r_e$.
\end{lemma}

\begin{proof}
We first observe that $q\in-\sign(x_{p0}^E)$ implies that $|C_q^E-\xi_{p0}^E|\geq r_e$,  and thus $S_q^E$, $\alpha_q$ and $\beta_q$ are well defined.
Without loss of generality, we assume that $\xi_{p0}^E\in \{0\} \times \R_{> 0}$ and $q=-1$. In particular, if $x_{p0}^E\neq 0$, the pursuer first needs to reach the $y$-axis before reaching $S_q^E$, which can only increase $\bar{t}_p$. Moreover, if $x_{p0}^E=0$ and $y_{p0}^E< 0$ then $S_q^E=0$, i.e., $\bar{t}_e=0$ and $\bar{t}_p\geq 0$.

We observe that $\frac{\beta_{-1}}{2}=\arctan(\frac{y_{p0}^E}{r_e})<\pi$ and $y_{p0}^E = |S_{-1}^E-\xi_{p0}^E|$. Due to the fact that $z\geq \arctan(z)$ for all $z\geq 0$ it holds that
    $\tfrac{y_{p0}^E}{r_e} \geq \tfrac{\beta_{-1}}{2}$.
Using $\bar{t}_p= \frac{y_{p0}^E}{v_p}$ and $\bar{t}_e=\frac{\beta_{-1} r_e}{v_e}$ leads to
\begin{align*}
    \frac{\bar{t}_p v_p}{r_e} = \frac{y_{p0}^E}{r_e} \geq  \frac{\beta_{-1}}{2} = \frac{\bar{t}_e v_e}{2 r_e}.
\end{align*}
This expression can be simplified to
\begin{align*}
    \frac{\bar{t}_p}{\bar{t}_e}  \geq \frac{v_e}{2 v_p} \geq 1
\end{align*}
under Assumption \ref{as:speed_ratio}, which completes the proof.
\end{proof}

To proceed, we summarize the optimal strategies of the evader and the pursuer under Assumption \ref{as:speed_ratio}.

\begin{lemma} \label{lem:optimal_inputs_solutions}
 Consider an evader \eqref{eq:cartesianDynamics} with $\xi_{e\theta0}^E =[0,0,\frac{\pi}{2}]^\top$ and a pursuer \eqref{eq:cartesianDynamicsPursuer} with $\xi_{p0}^E\in \R^2$.
 Under Assumption \ref{as:speed_ratio} and $q\in -\sign(x_{p0}^E)$ the optimal strategies of the evader and the pursuer to maximize/minimize the distance between each other is given by
    \begin{align}
        u_e^{*_q}(t)&\in \left\{ \begin{array}{cc}
           q  &  \text{for } t< \bar{t}_e \\
           0  & \text{for } t \geq \bar{t}_e
        \end{array} \right. \label{eq:u_e*} \\
        u_p^{*_q}(t) &= \left\{ \begin{array}{cc}
         \arctan_2(S_q^E-\xi_{p0}^E) & \text{for } \xi_{p0}^E \neq 0  \\
             \frac{\pi}{2} & \text{for } \xi_{p0}^E =0
        \end{array} \right. \label{eq:u_p*}
    \end{align}
    and where $S_q^E$ and $\bar{t}_e$ are defined in \eqref{eq:def_S} and \eqref{eq:def_t_e}, respectively.
    The corresponding solutions in the $(x,y)$-plane satisfy
    \begin{align}
        \xi_e(t,u_e^{*_q}) &\! = \! \left\{ \! \! \!\! \begin{array}{cl}
             C_{q}^E -q r_e R(t \frac{v_e}{r_e})^{-q} \left[ \begin{smallmatrix}
         1  \\
         0 
    \end{smallmatrix}\right]  & \! \text{if } t\leq \bar{t}_e \\
            S_q^E \! +\! (t\!- \! \bar{t}_e) v_e \tfrac{S_q^E-\xi_{p0}^E}{|S_q^E-\xi_{p0}^E|} & \! \text{if } t\geq \bar{t}_e
        \end{array}  \right. \\
        \xi_p(t,u_p^{*_q}) & \! = \! \xi_{p0}^E+t v_p \tfrac{S_q^E-\xi_{p0}^E}{|S_q^E-\xi_{p0}^E|}. \label{eq:solution_xi_p1v1}
    \end{align}
\end{lemma}

\begin{proof}
The result follows immediately from the results derived in \cite{braun_prying_2025} summarized in Proposition  \ref{prop:main_result_journal1v1}, Lemma \ref{lem:illustration_speed_ratio_assumption} and the (geometric) calculations in this section.
\end{proof}

\begin{remark} \label{rem:suboptimal}
    While the optimal inputs in Lemma \ref{lem:optimal_inputs_solutions} rely on the assumption $q\in -\sign(x_{p0}^E)$, if $|C_{q}^E-\xi_{p0}^E|\geq r_e$, the definitions and solution representations 
    \eqref{eq:u_e*}-\eqref{eq:solution_xi_p1v1}, (even though they are not optimal,)
    are also valid if $q\notin -\sign(x_{p0}^E)$. Accordingly, \eqref{eq:u_e*}-\eqref{eq:solution_xi_p1v1} will be used for $q\in \{-1,1\}$ in the following. Moreover, observe that Lemma \ref{lem:optimal_inputs_solutions} is independent of the selection of the surveillance radius $\rho$ and thus, \eqref{eq:u_e*}-\eqref{eq:solution_xi_p1v1} are defined for all $t\in \R_{\geq 0}$.
\end{remark}

Based on the assumptions and definitions in Lemma \ref{lem:optimal_inputs_solutions}, the distance between the evader and the pursuer is given by\footnote{Since $\xi_{e\theta0}^E =[0,0,\frac{\pi}{2}]^\top$ by assumption, the distance is only a function of the pursuers initial position.}
\begin{align}
    &d_q(t;\xi_{p0}^E) \\
    &\!= \! \left\{ \! \! \! \begin{array}{cc}
     \left| 
    C_{q}^E -\xi_{p0}^E -q r_e R(t \frac{v_e}{r_e})^{-q} \left[ \begin{smallmatrix}
         1  \\
         0 
    \end{smallmatrix}\right] -t v_p \tfrac{S_q^E-\xi_{p0}^E}{|S_q^E-\xi_{p0}^E|}  \right|, & \!\!  t \leq \bar{t}_e \\
     |S_q^E\! + \! (t(v_e\! -\! v_p)\! -\! \bar{t}_e v_e)\tfrac{S_q^E-\xi_{p0}^E}{|S_q^E-\xi_{p0}^E|} + C_{q}^E-\xi_{p0}^E|,    & \! \! t \geq \bar{t}_e
    \end{array} 
     \right. \nonumber
\end{align}

The optimal value function satisfies $T^{*_q}=V(\xi_{e\theta 0}^E, \xi_{{p}0}^E)$, $q\in -\sign(x_{p0}^E)$, where $T^{*_q}$ is defined through the condition
\begin{align}
    \rho = d(T^{*_q},\xi_{p0}^E). \label{eq:T*}
\end{align}
Since the pursuer is following a straight line, we can also consider the pursuer to be stationary and define
\begin{align}
    &\bar{d}_q(t;\xi_{p0}^E) \\
    &= \left\{ \! \! \! \begin{array}{cc}
     \left| 
    C_{q}^E -\xi_{p0}^E -q r_e R(t \frac{v_e}{r_e})^{-q} \left[ \begin{smallmatrix}
         1  \\
         0 
    \end{smallmatrix}\right]  \right|, &  t \leq \bar{t}_e \\
     |S_q^E+(t-\bar{t}_e) v_e\tfrac{S_q^E-\xi_{p0}^E}{|S_q^E-\xi_{p0}^E|} + C_{q}^E-\xi_{p0}^E|,    & t \geq \bar{t}_e
    \end{array} 
     \right. \nonumber
\end{align}
which characterizes the distance of the evader to the initial position $\xi_{p0}$, i.e.,
in this case $T^{*_q}$ is defined through
   $ \rho + T^{*_q} v_p = \overline{d}_q(T^{*_q},\xi_{p0}^E)$.

\section{Extensions to the 2v1 setting} 
\label{sec:2v1}

In this section we discuss extensions of the results in the previous section to the 2v1 setting under the additional assumption that the pursuers are stationary. Before we derive optimal strategies, we introduce definitions and present an alternative coordinate system.

\subsection{Definitions, coordinate transformations \& first results}

To discuss the 2v1 setting in Problem \ref{prob:game_of_degree}, we consider a coordinate system defined through the initial positions of the two pursuers. Without loss of generality, we assume that
\begin{align}
    \xi_{p_10}^P= \left[ \begin{array}{c}
         -\chi  \\
         0 
    \end{array}\right], \qquad     
    \xi_{p_20}^P= \left[ \begin{array}{c}
         \chi  \\
         0 
    \end{array}\right],
\end{align}
and where $2\chi \geq 0$ defines the distance between the two pursuers. In this coordinate system, the initial position of the evader is given by $\xi_{e\theta0}^P \in\R^2\times (-\pi,\pi]$.
For $\rho\geq \chi$ we define 
\begin{align}
    H_j^P = \left[ \begin{array}{c}
         0  \\
         j \sqrt{\rho^2- \chi^2} 
    \end{array}\right], \qquad j\in \{-1,1\}, \label{eq:HP}
\end{align}
visualized in Fig.~\ref{fig:2v1_setting}. The points denote the intersection of circles with radius $\rho$ around the pursuers.
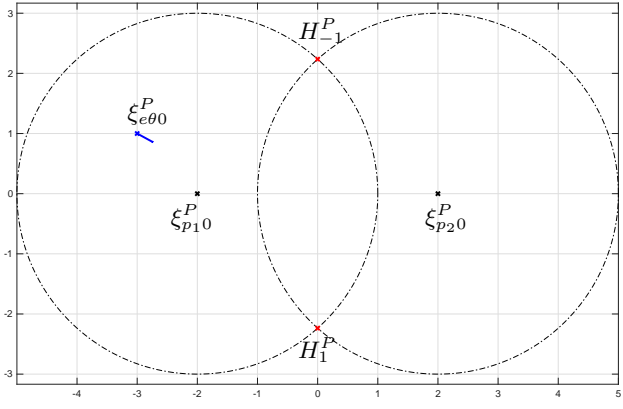
\begin{figure}[htb]
    \centering
    \begin{overpic}[width = 0.95\columnwidth]{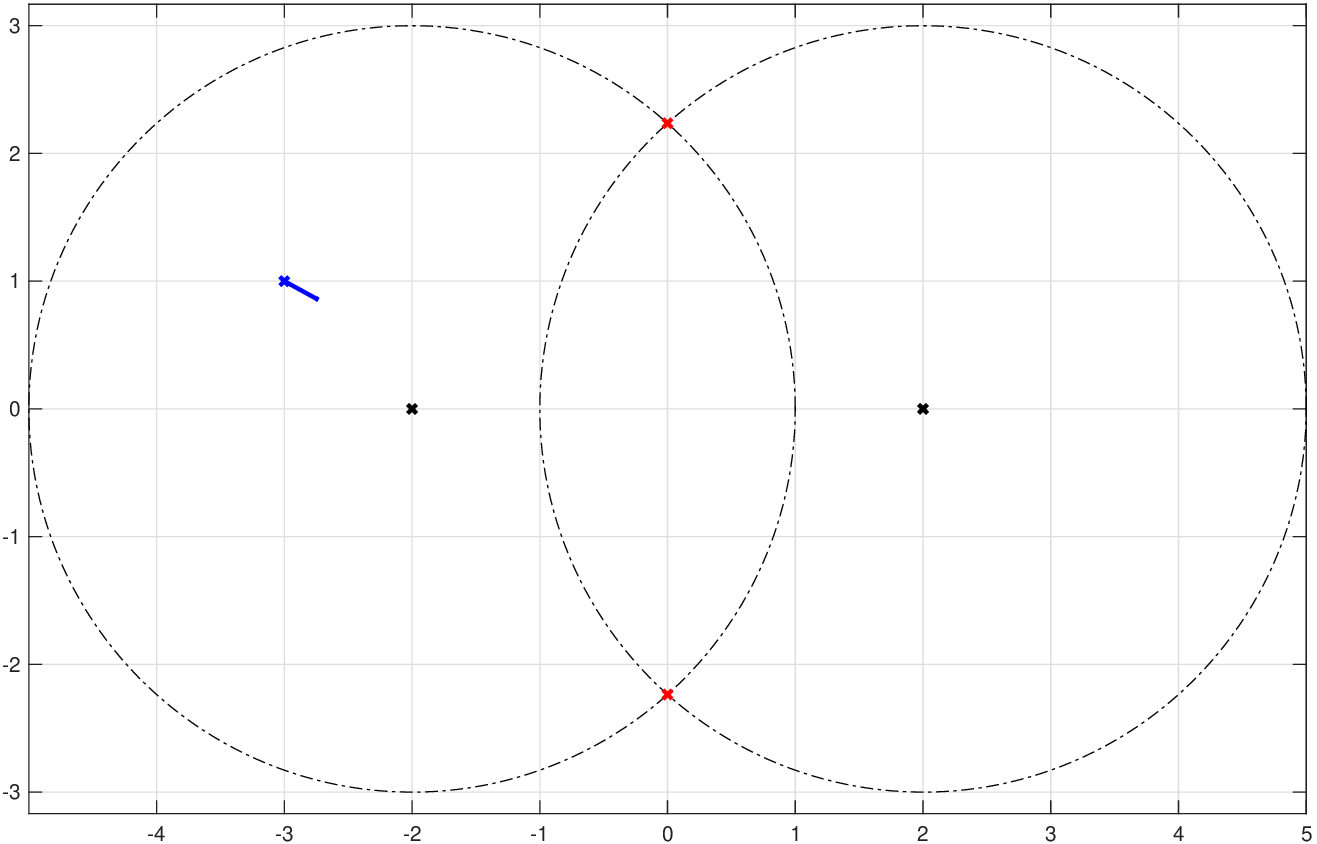}
    \put(20,45){\small{$\xi_{e\theta 0}^P$}}
    \put(47.5,6.5){\small{$H_1^P$}}
    \put(47.5,58){\small{$H_{-1}^P$}}
    \put(27,28){\small{$\xi_{p_10}^P$}}
    \put(68,28){\small{$\xi_{p_20}^P$}}
    \end{overpic}
    \caption{Coordinate system defined through the initial positions of the pursuers.}
    \label{fig:2v1_setting}
\end{figure}

For $c> 0$, $i\in\{1,2\}$, $s\in \{1,2\}\backslash \{i\}$, we define the sets
\begin{align}
\begin{split}
    \mathcal{P}_i^P(c) &= \{\xi \in \R^2| \ |\xi-\xi_{p_i0}^P|=c, \ |\xi-\xi_{p_s0}^P|> c\}, \\ 
    \bar{\mathcal{P}}_i^P(c) &=  \{\xi \in \R^2| \ |\xi-\xi_{p_i0}^P|=c, \ |\xi-\xi_{p_s0}^P|\leq c\}. 
\end{split} \label{eq:terminatiing_sets}
\end{align}
Here $\mathcal{P}_i^P(\rho)$, $i\in \{1,2\}$, denotes the set where the game of degree can end, while on $\bar{\mathcal{P}}_i^P(\rho)$, $i\in \{1,2\}$, the evader can potentially leave the surveillance radius of one pursuer but not the surveillance radius of the other pursuer.

To switch between the two coordinate systems denoted through $E$ and $P$ we define
\begin{align}
    \xi^E &= R(\theta-\tfrac{\pi}{2})^\top(\xi^P- \xi_{e0}^P), \qquad \xi^P \in \R^2, \label{eq:co_trans_P_to_E} \\
    \xi^P &= R(\theta-\tfrac{\pi}{2}) \xi^E+\xi_{e0}^P , \ \, \quad \qquad \xi^E \in \R^2, \label{eq:co_trans_E_to_P}
\end{align}
to go from $P$ to $E$ and from $E$ to $P$, respectively.
This in particular allows us to calculate
\begin{alignat*}{3}
    \xi_{p_i0}^E &= R(\theta-\tfrac{\pi}{2})^\top(\xi_{p_i0}^P- \xi_{e0}^P), \qquad && i\in \{1,2\}, \\
    H_{q}^E &= R(\theta-\tfrac{\pi}{2})^\top(H_{q}^P- \xi_{e0}^P), \qquad && q\in \{-1,1\}, \\
        S_q^P &= R(\theta-\tfrac{\pi}{2}) S_q^E+\xi_{e0}^P, \qquad && q\in \{-1,1\}, \\
            C_q^P &= R(\theta-\tfrac{\pi}{2}) C_q^E+\xi_{e0}^P, \qquad && q\in \{-1,1\}. 
\end{alignat*}
to switch between the different coordinate systems for the points highlighted in Figs. \ref{fig:1v1_setting} and \ref{fig:2v1_setting}.
Since we can easily switch between the two coordinate systems, we can use the strategies in Lemma \ref{lem:optimal_inputs_solutions} to state the following result.

\begin{lemma} \label{lem:strategy_dominant}
    Consider  
    Problem \ref{prob:game_of_degree} with $v_{p_1}=v_{p_2}=0$.
    Let $\xi_{e0}^P\in \mathcal{B}_{\rho}(\xi_{p_10}^P)\cup \mathcal{B}_{\rho}(\xi_{p_20}^P)$, define 
$q_i\in -\sign(x_{p_i0}^E)$ for $i\in \{1,2\}$ and let 
$u_e^{*_{q_i}}(t)$ denote the optimal control strategies from the 1v1 setting  defined in
\eqref{eq:u_e*} with corresponding optimal values $T^i = V(\xi_{e\theta 0}^P,\xi_{{p_i}0}^P)$.

If there exists $i\in \{1,2\}$ such that $\xi_{e}^P(T^i)\in \mathcal{P}_i^P(\rho)$ and $T^i>0$, then $u_e^{*_{q_i}}(t)$ is optimal for the 2v1 setting.
\end{lemma}

\begin{proof}
Let $i\in \{1,2\}$ such that   $\xi_{e}^P(T^i)\in \mathcal{P}_i^P(\rho)$ which exists according to the assumptions. Since $u_e^{*_{q_i}}(t)$ is optimal for the 1v1 setting the 2v1 setting cannot end before $T^i$, i.e., $T^i$ provides a lower bound for the 2v1 setting. Since $\xi_{e}^P(T^i)\in \mathcal{P}_i^P(\rho)$ and $T^i>0$ the 2v1 game ends at time $T^i$, i.e., $u_e^{*_{q_i}}(t)$ is optimal for the 2v1 game.
\end{proof}

A setting where the assumptions of Lemma \ref{lem:strategy_dominant} are satisfied is shown in Fig.~\ref{fig:2v1_dominant}.
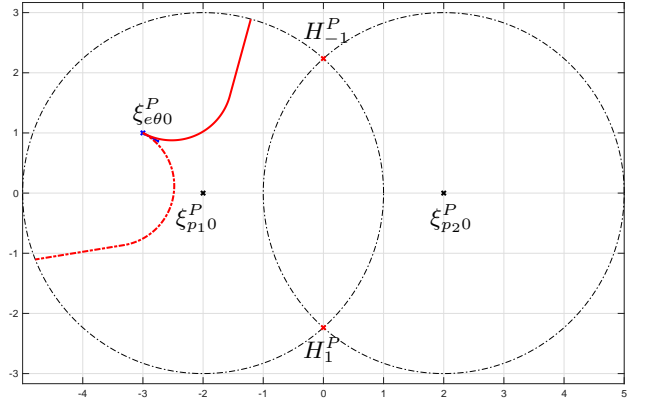
\begin{figure}[htb]
    \centering
    \begin{overpic}[width = 0.95\columnwidth]{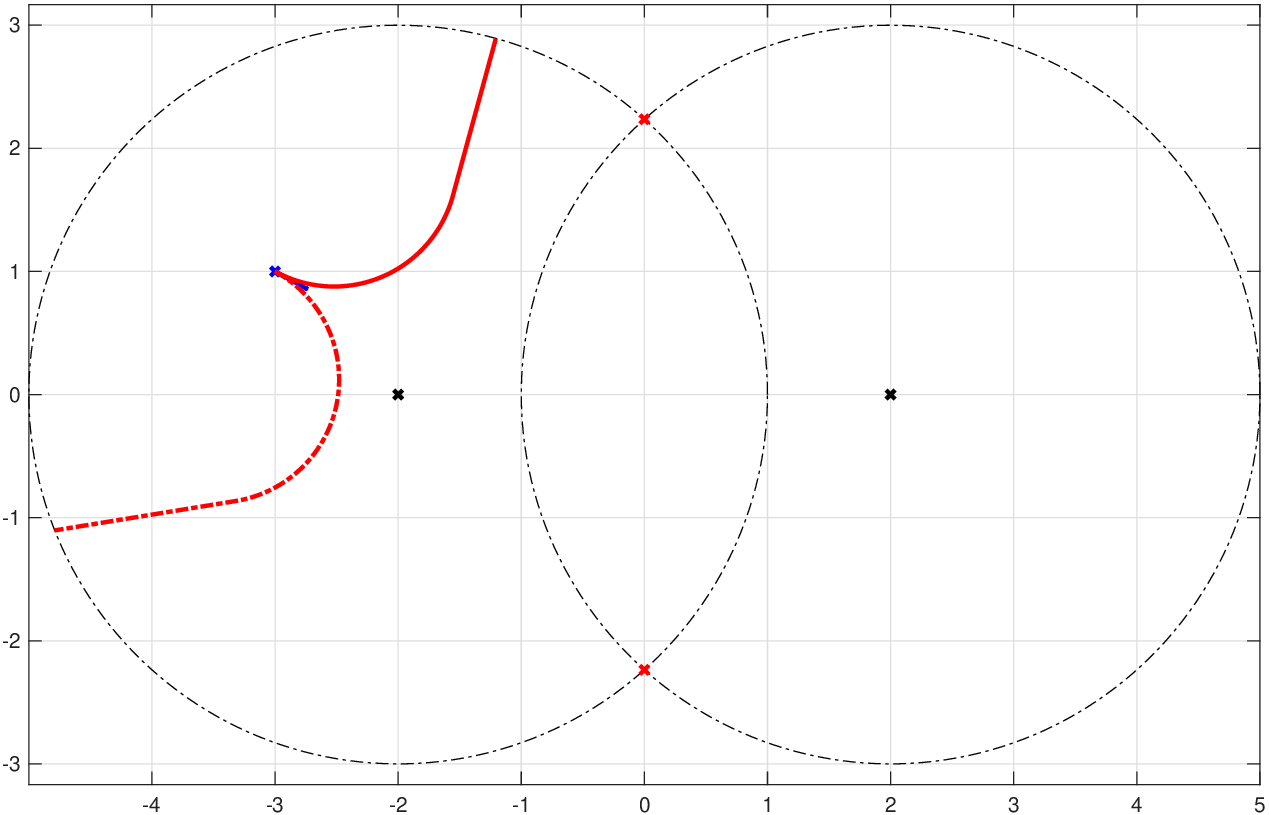}
    \put(20,45){\small{$\xi_{e\theta 0}^P$}}
    \put(47.5,6.5){\small{$H_1^P$}}
    \put(47.5,58){\small{$H_{-1}^P$}}
    \put(27,28){\small{$\xi_{p_10}^P$}}
    \put(68,28){\small{$\xi_{p_20}^P$}}
    \end{overpic}
    \caption{Visualization of the setting where the assumptions in Lemma \ref{lem:strategy_dominant} are satisfied. The optimal strategy with respect to pursuer 1 is shown in a solid line while the optimal strategy with respect to pursuer 2 is shown as a dashed line. From Lemma \ref{lem:strategy_dominant} we can conclude that the strategy focusing on pursuer 1 is optimal.}
    \label{fig:2v1_dominant}
\end{figure}
The case where $\xi_{e}^P(T^i)\notin \mathcal{P}_i^P(\rho)$ for $i\in \{1,2\}$ is discussed in the next section.

\subsection{Shortest path to a point \& (sub)optimal strategies}

The optimal control strategies from the 1v1 setting are not sufficient to describe the optimal strategies for the 2v1 setting. The 1v1 inputs $u_e^{*_{q_1}}(t)$ and $u_e^{*_{q_2}}(t)$, $q\in -\sign(x_{p_i0}^E)$, may lead to trajectories such that $\xi_e^P(T^{1})\in \bar{\mathcal{P}}_1^P(\rho)$ and $\xi_e^P(T^{2})\in \bar{\mathcal{P}}_2^P(\rho)$. In this case it might be better to aim for the points $H_j^P$, $j\in \{-1,1\}$, to terminate the game.

In this section we derive the general shortest path from $\xi_{e\theta0}^E$ to a point $\kappa^E \in \R^2$ (without final orientation) satisfying the dynamics \eqref{eq:cartesianDynamics} with constant speed $v_e$. In this case, the shortest path and the corresponding maneuver are characterized through a turn-straight or a turn-turn maneuver, where each component might be of length zero \cite[Cor. 2]{braun_capture_2023}, \cite{dubins_curves_1957}.

For the derivation we again use the coordinate system $E$ in Fig.~\ref{fig:1v1_setting}, providing a simple representation of the initial position $\xi_{e\theta0}^E$ and the turning points $C_q^E$, $q\in \{-1,1\}$.
We define $q\in \sign(\kappa^E(1))$, where $\kappa^E(1)$ denotes the first component of $\kappa^E=[\kappa^E(1),\kappa^E(2)]^\top$, and we assume $\kappa^E\notin \mathcal{B}_{r_e}(C_q^E)$, i.e., we turn around $C_q^E$, where $q$ depends on the sign of the $x$-component of the target point and the distance between $\kappa^E$ and $C_q^E$ is larger or equal to $r_e$ (see \eqref{eq:turning_radius}). The setting is visualized in Fig.~\ref{fig:1v1_setting_bar}.
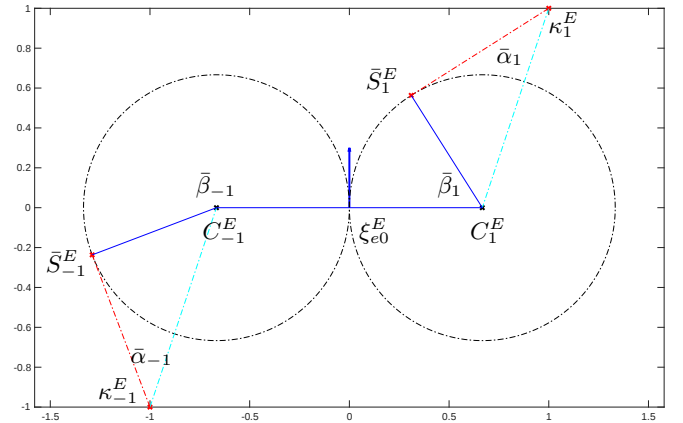
\begin{figure}[b]
    \centering
    \begin{overpic}[width = 1\columnwidth]{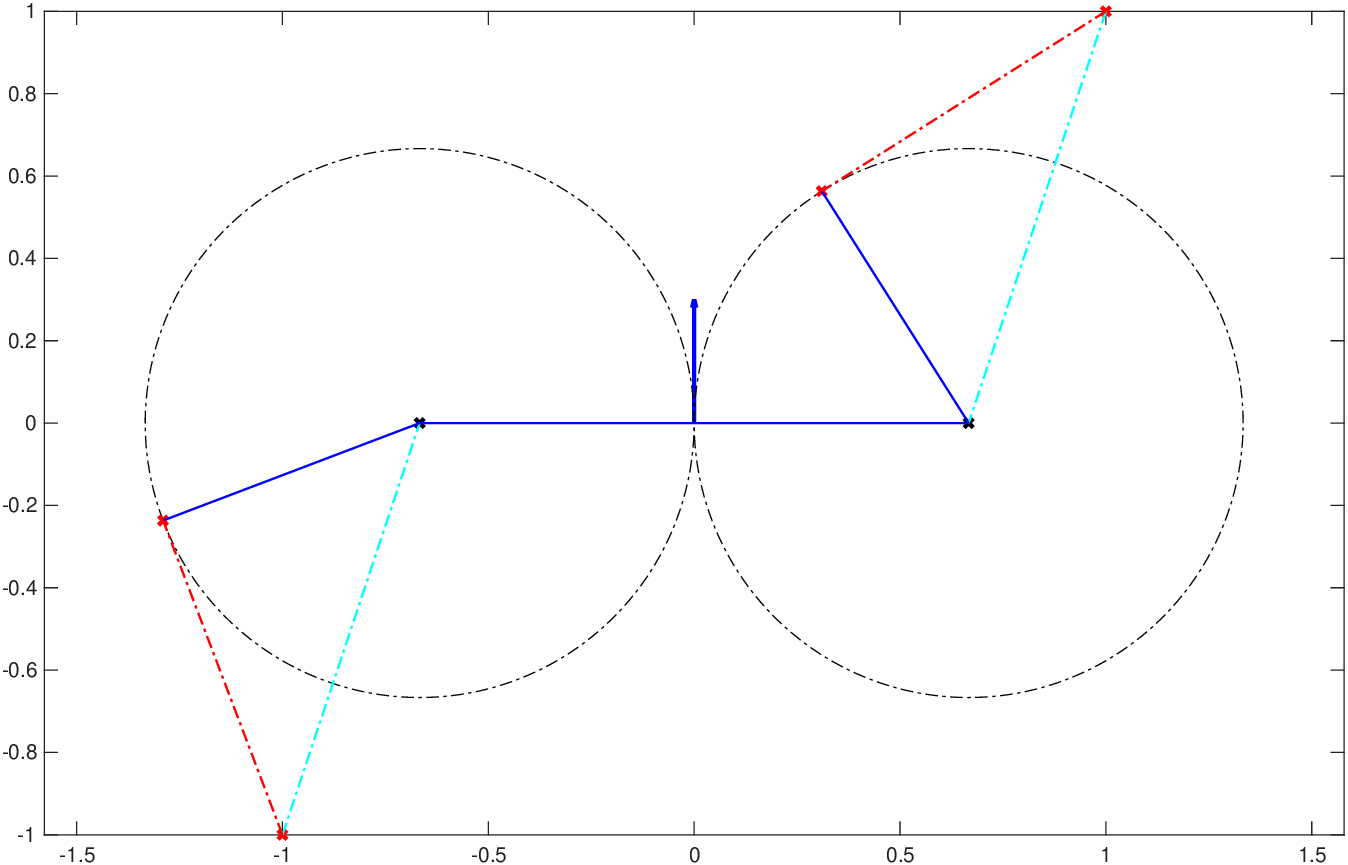}
    \put(53,28){\small{$\xi_{e0}^E$}}
    \put(70,28){\small{$C_1^E$}}
    \put(29,28){\small{$C_{-1}^E$}}
    \put(5,23){\small{$\bar{S}_{-1}^E$}}
    \put(82,59.5){\small{$\kappa_1^E$}}
    \put(13,3.5){\small{$\kappa_{-1}^E$}}
    \put(18,9){\small{$\bar{\alpha}_{-1}$}}
    \put(28,35){\small{$\bar{\beta}_{-1}$}}
    \put(74,55){\small{$\bar{\alpha}_{1}$}}
    \put(65,35){\small{$\bar{\beta}_{1}$}}
    \put(54,51){\small{$\bar{S}_{1}^E$}}
    \end{overpic}
    \caption{Shortest path to a point $\kappa_q^E$, characterized through a turn-straight maneuver.}
    \label{fig:1v1_setting_bar}
\end{figure}
Under these assumptions, the shortest path is given by a turn-straight maneuver (where both segments can be of length zero) and the case turn-turn can be ruled out.
The derivation of the parameters in Fig.~\ref{fig:1v1_setting_bar} follow similar steps as in Section \ref{sec:1v1_geometric_solution}. In particular, for $q\in \{-1,1\}$, we define
\begin{align}
    \bar{\alpha}_q &= \arcsin\left( \frac{r_e}{|C_{q}^E-\kappa^E|} \right), \label{eq:alpha_bar} \\
    \bar{S}_q^E &= \frac{|\bar{S}_q^E-\kappa^E|}{|C_{q}^E-\kappa^E|} R(\alpha_q)^{-q} (C_{q}^E-\kappa^E) +\kappa^E. \label{eq:def_S_bar} 
\end{align}
and where $|\bar{S}_q^E-\kappa^E|$ can be obtained from 
\begin{align*}
    |\bar{S}_q^E-\kappa^E| = \sqrt{|C_{q}^E-\kappa^E|^2 - r_e^2}.
\end{align*}
Here, compared with \eqref{eq:def_S}, the inverse of the rotation matrix is used.
The angle $\bar{\beta}_q$ characterizing the arc from $\xi_{e0}^E$ to $\bar{S}_q^E$ is given by
\begin{align}
    \bar{\beta}_q &= \left\{ \! \!\! \! \begin{array}{cc}
         \arccos \left( \tfrac{r_e^2- (\bar{S}_q^E)^\top C_{q}^E}{r_e^2}    \right) & \text{if } \bar{S}_q^E(2) \geq 0  \\
         2\pi \! - \! \arccos \left( \! \tfrac{r_e^2- (\bar{S}_q^E)^\top C_{q}^E}{r_e^2}    \right) &  \text{if } \bar{S}_q^E(2) \leq 0
    \end{array} \right. \label{eq:beta_bar}
\end{align}
similar to \eqref{eq:beta}.

The calculations in this section can be applied to characterize the shortest path to $H_j^P$, $j\in \{-1,1\}$, and where we can again easily switch between $H_j^P$ and $H_j^E$ through the transformations in \eqref{eq:co_trans_P_to_E}, \eqref{eq:co_trans_E_to_P}.

\begin{lemma} \label{lem:optimal_strategy_to_point}
Consider \eqref{eq:cartesianDynamics} and $H_j^E\in \R^2$, $j\in \{-1,1\}$, defined in \eqref{eq:HP}. Let $q\in \sign(H_j^E(1))$ and assume that  $H_j^E\notin \mathcal{B}_{r_e}(C_q^E)$, $r_e>0$, and $C_q^E\in \R^2$ (see \eqref{eq:C_q}). Then the shortest path from $\xi_{e0}$ to $H_j^E$ satisfying the dynamics \eqref{eq:cartesianDynamics} is characterized through the input and the trajectory
    \begin{align}
        u_e^{H_j^E}(t)&\in \left\{ \begin{array}{cl}
           q  &  \text{if } t< \frac{\bar{\beta}_q}{v_e}r_e \\
           0  & \text{if } t \geq \frac{\bar{\beta}_q}{v_e}r_e
        \end{array} \right. \label{eq:u_eH} \\
\! \!        \xi_e(t,u_e^{H_j^E}) &\! = \! \left\{ \! \! \!\! \begin{array}{cl}
             C_{q}^E -q r_e R(t \frac{v_e}{r_e})^{-q} \left[ \begin{smallmatrix}
         1  \\
         0 
    \end{smallmatrix}\right]  &  \! \text{if } t\leq \frac{\bar{\beta}_q}{v_e} r_e \\
            \bar{S}_q^E \! +\! (t\!- \! \frac{\bar{\beta}_q}{v_e}) v_e \tfrac{H_j^E-\bar{S}_q^E}{|H_j^E-\bar{S}_q^E|} &  \! \text{if } t\geq \frac{\bar{\beta}_q}{v_e} r_e
        \end{array}  \right.  \label{eq:opt_H_j_e}
    \end{align}
where $\bar{S}_q^E$ and $\bar{\beta}_q$
are defined in
\eqref{eq:def_S_bar} and 
\eqref{eq:beta_bar}, respectively.
\end{lemma}

\begin{proof}
Follows from the definitions in this section and the characterization of Dubin's paths in \cite[Cor. 2]{braun_capture_2023}, \cite{dubins_curves_1957}. Specifically, since $H_j^E\notin \mathcal{B}_{r_e}(C_q^E)$ by assumption, the point $H_j^E$ is reached through a turn-straight maneuver (instead of a turn-turn maneuver). The turn and the straight maneuver are characterized through the two cases in \eqref{eq:u_eH} and \eqref{eq:opt_H_j_e}, and the two cases are constructed similar to Lemma \ref{lem:optimal_inputs_solutions}  with $\frac{\bar{\beta}_q}{v_e}r_e$ denoting the time when the control is switched.
\end{proof}

\begin{remark}
    If $v_p>0$, the optimal pursuer strategy $u_p^{H_j^E}$ and trajectory corresponding to Lemma \ref{lem:optimal_strategy_to_point} are given by
    \begin{align}
        u_p^{H_j^E}(t) &= \left\{ \begin{array}{cc}
         \arctan_2(\bar{S}_q^E-\xi_{p0}^E) & \text{for } \xi_{p0}^E \neq 0  \\
             \frac{\pi}{2} & \text{for } \xi_{p0}^E =0
        \end{array} \right. \label{eq:u_pH} \\
        \xi_p(t,u_p^{H_j^E}) &  =  \xi_{p0}^E+t v_p \tfrac{\bar{S}_q^E-\xi_{p0}^E}{|\bar{S}_q^E-\xi_{p0}^E|}. \label{eq:_xi_pH}
    \end{align}
\end{remark}

With Lemma \ref{lem:optimal_strategy_to_point}, we are ready to state the next result, capturing in which cases the evader should aim for $H_j^E$ to terminate the game. The setting is visualized in Fig.~\ref{fig:2v1_no_dominant}.
\begin{figure}[htb]
    \centering
    \begin{overpic}[width = 1\columnwidth]{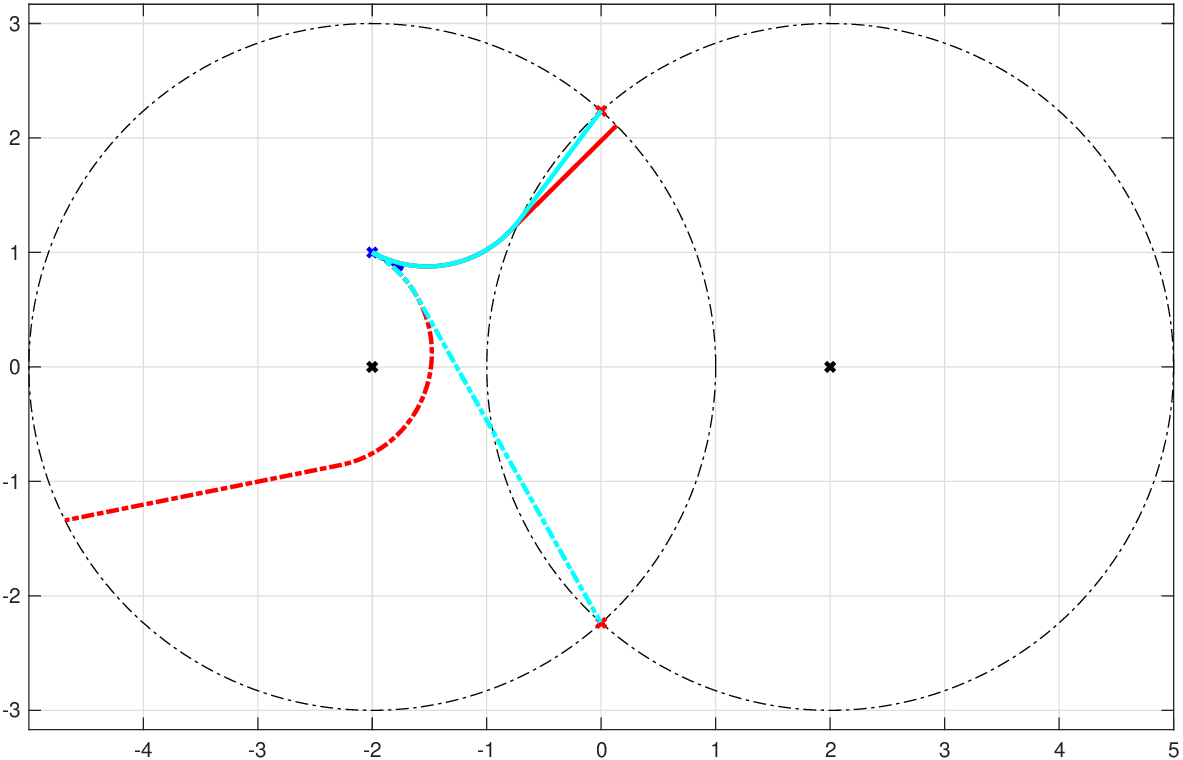}
    \put(28,45){\small{$\xi_{e\theta 0}^P$}}
    \put(47.5,6.5){\small{$H_1^P$}}
    \put(47.5,58){\small{$H_{-1}^P$}}
    \put(27,28){\small{$\xi_{p_10}^P$}}
    \put(68,28){\small{$\xi_{p_20}^P$}}
    \end{overpic}
    \caption{Visualization of the setting where the assumptions in Lemma \ref{lem:strategy_dominant} are not satisfied. Instead, the evader has to aim for the points $H_j^{P}$, $j\in \{-1,1\}$ as outlined in Lemma \ref{lem:strategy_same_sides}.}
    \label{fig:2v1_no_dominant}
\end{figure}
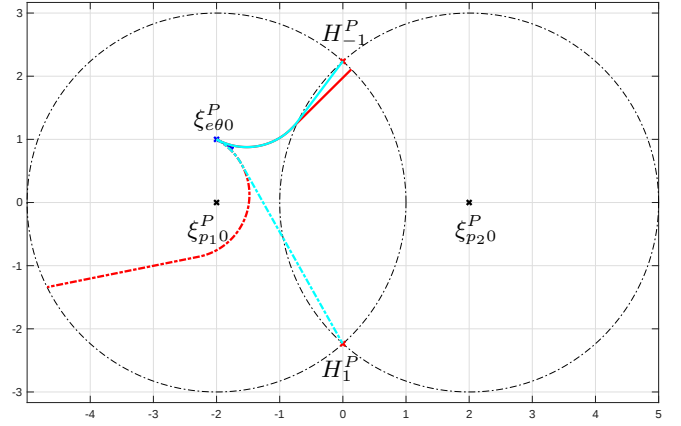

\begin{lemma} \label{lem:strategy_same_sides}
Consider Problem \ref{prob:game_of_degree} with $v_{p_1}=v_{p_2}=0$.
    Let $\xi_{e0}^P\in \mathcal{B}_{\rho}(\xi_{p_10}^P)\cup \mathcal{B}_{\rho}(\xi_{p_20}^P)$, define 
$q_i\in -\sign(x_{p_i0}^E)$ for $i\in \{1,2\}$ and let 
$u_e^{*_{q_i}}(t)$ denote the optimal control strategies from the 1v1 setting  defined in
\eqref{eq:u_e*} with corresponding optimal values $T^i = V(\xi_{e\theta 0}^P,\{\xi_{{p_i}0}^P\}_{i=1}^2)$.
Assume that $\xi_{e}^P(T^i)\in \bar{\mathcal{P}}_i^P(\rho)$ for all $i\in \{1,2\}$, 
for $j\in \{-1,1\}$ define $q_j\in \sign(H_{j}^E(1))$ and assume that  $H_{q_j}^E\notin \mathcal{B}_{r_e}(C_{q_j}^E)$ for all $j\in \{-1,1\}$.
For $j\in \{-1,1\}$ consider 
$u_e^{H_j^E}(t)$ and $\xi_e^E(t,u_e^{H_j^E})$ (see 
\eqref{eq:u_eH}, \eqref{eq:opt_H_j_e}), let $\xi_e(T^{H_j},u_e^{H_j^E})=H_j^E$ and 
$
    k= \argmin_{j\in \{-1,1\}} T^{H_j}.
$

Then $u_e^{H_k^E}(t)$ is the optimal evader strategy and $T^{H_k}$ is the minimal time to reach the set  $\overline{\mathcal{P}_1^P(\rho) \cup \mathcal{P}_2^P(\rho)}$.
\end{lemma}

\begin{proof}
    The condition $H_{q_j}^E\notin \mathcal{B}_{r_e}(C_{q_j}^E)$ ensures that the components in \eqref{eq:u_eH}, \eqref{eq:opt_H_j_e} are well defined and $H_{q_j}^E$ can be reached in a turn-straight maneuver. This condition in combination with the assumption $\xi_{e}^P(T^i)\in \bar{\mathcal{P}}_i^P(\rho)$ for all $i\in \{1,2\}$,
    additionally implies that $T^{H_j}>0$.

    Since $\xi_{e}^P(T^i)\in \bar{\mathcal{P}}_i^P(\rho)$, $i\in \{1,2\}$, strategies from the 1v1 setting are not optimal. However, the closest point from $\xi_{e}^P(T^i)$, $i\in \{1,2\}$, to $\overline{\mathcal{P}_1^P(\rho) \cup \mathcal{P}_2^P(\rho)}$, to terminate the game is contained in $\{H_{q_{-1}}^E,H_{q_{1}}^E\}$. Thus, one of the strategies $u_e^{H_j^E}(t)$, $j\in \{-1,1\}$, is optimal.
\end{proof}

\begin{remark}
Note that we do not claim that $u_e^{H_k^E}(t)$ is optimal to terminate the game. Here, we might not reach $H_{q_{k}}$ with the correct angle pointing outside of the game set.
\end{remark}

In the case that $H_{q_j}^E \in \mathcal{B}_{r_e}(C_{q_j}^E)$ for $j\in \{-1,1\}$ an additional strategy might be necessary. In this case it is not possible to reach $H_{q_j}^E$ through a single turn-straight maneuver. The setting is visualized in Fig.~\ref{fig:2v1_no_dominant_noH}, where $H_{-1}^P$ cannot be reached through a simple left turn.
In this case we use the strategy \eqref{eq:u_e*}, but define $q$ differently. We conclude this section with the following conjecture.

\begin{figure}[htb]
    \centering
    \begin{overpic}[width = 1\columnwidth]{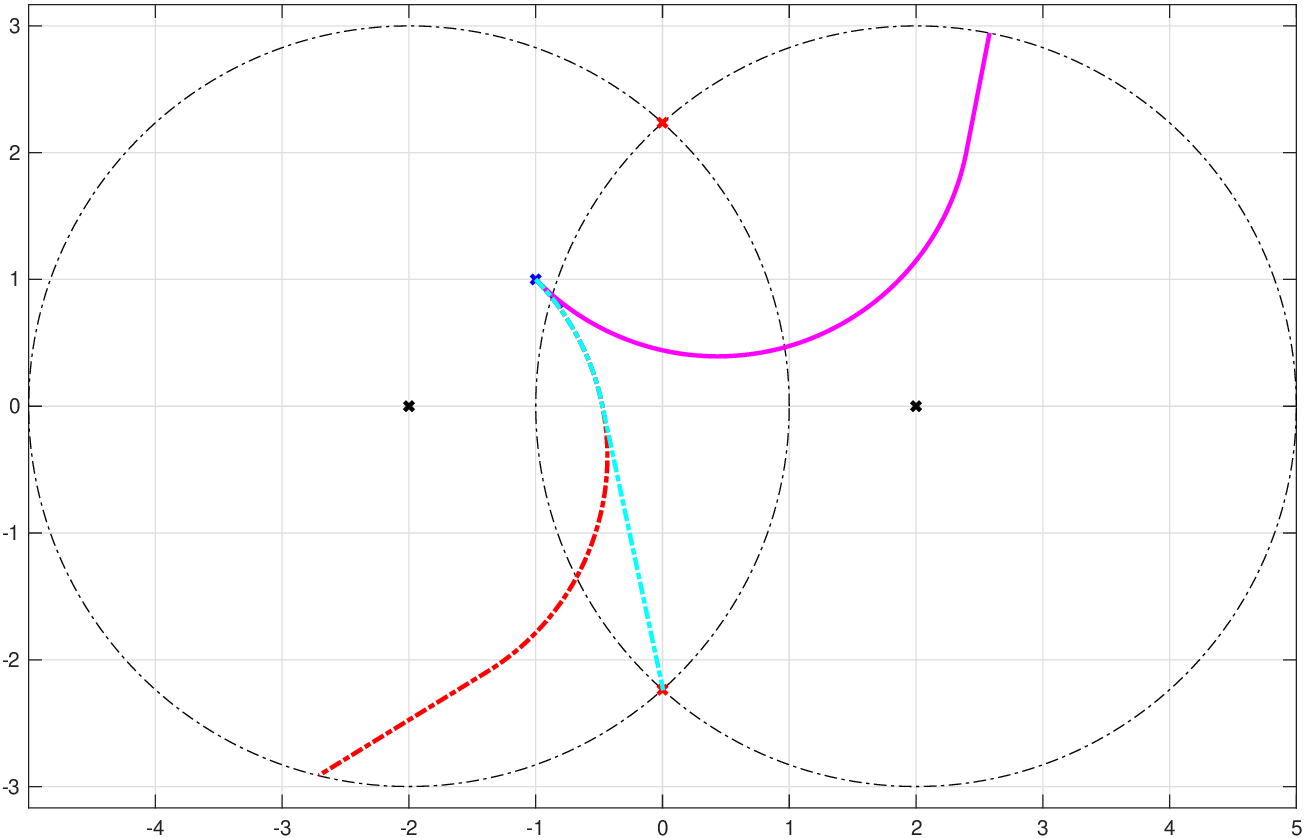}
    \put(36,45){\small{$\xi_{e\theta 0}^P$}}
    \put(47.5,6.5){\small{$H_1^P$}}
    \put(47.5,58){\small{$H_{-1}^P$}}
    \put(27,28){\small{$\xi_{p_10}^P$}}
    \put(68,28){\small{$\xi_{p_20}^P$}}
    \end{overpic}
    \caption{Illustration of the assumptions covered in Conjecture \ref{conj:strategy_same_sides}. It is not possible to reach $H_{-1}^P$ through a turn-straight maneuver. In this case we consider the trajectories to $H_1^P$ (cyan) and the trajectory in magenta (see Remark \ref{rem:suboptimal}) as candidates for the optimal solution.}
    \label{fig:2v1_no_dominant_noH}
\end{figure}
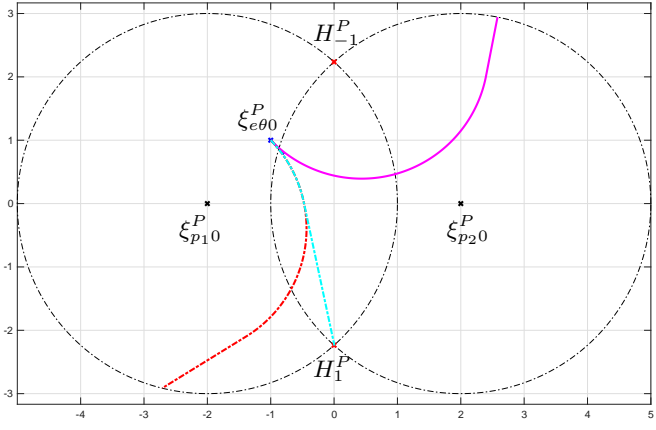

\begin{cconj} \label{conj:strategy_same_sides}
Consider Problem \ref{prob:game_of_degree} with $v_{p_1}=v_{p_2}=0$, let $\xi_{e0}^P\in \mathcal{B}_{\rho}(\xi_{p_10}^P)\cup \mathcal{B}_{\rho}(\xi_{p_20}^P)$, define 
$q_i\in -\sign(x_{p_i0}^E)$ for $i\in \{1,2\}$ and let 
$u_e^{*_{q_i}}(t)$ denote the 1v1 strategies defined in
\eqref{eq:u_e*} with corresponding $T^i = V(\xi_{e\theta 0}^P,\{\xi_{{p_i}0}^P\}_{i=1}^2)$.
Assume that $\xi_{e}^P(T^i)\in \bar{\mathcal{P}}_i^P(\rho)$ for all $i\in \{1,2\}$.
For $j\in \{-1,1\}$, define $q_j\in \sign(H_{j}^E(1))$. If $H_{q_j}^E \in \mathcal{B}_{r_e}(C_{q_j}^E)$ define
\begin{align}
    i=\left\{ \begin{array}{cl}
        1 & \text{if } q_j\cdot j= -1  \\
        2 & \text{if } q_j\cdot j= 1
    \end{array} \right.
\end{align}
and for $j\in \{-1,1\}$, if $H_{q_j}^E \in \mathcal{B}_{r_e}(C_{q_j}^E)$, define $u_e^j = u_e^{*_i}$ and $T^j$ through the condition $\xi_e^P(T^j,u_e^{*_i}) \in \mathcal{P}_i^P(\rho)$ where $u_e^{*_i}$ is defined in \eqref{eq:u_e*}.
If $H_{q_j}^E \notin \mathcal{B}_{r_e}(C_{q_j}^E)$ define
$u_e^j=u_e^{H_j^E}(t)$ (see 
\eqref{eq:u_eH}) and let  $T^j$ be defined through the condition $\xi_e^E(T^{j},u_e^{H_j^E})=H_j^E$.

If $k$ is defined as
$k=\argmin_{j\in \{-1,1\}} T^{j}$,
then $u_e^{k}(t)$ and $T^{k}$ are optimal for the 2v1 setting.
\end{cconj}

The setting covered in Conjecture \ref{conj:strategy_same_sides} is visualized in Fig.~\ref{fig:2v1_no_dominant_noH}. Focusing on $\xi_{p_20}^P$, the optimal strategy of the 1v1 setting is shown in red. However, since this strategy ends on $\bar{\mathcal{P}}_2^P(\rho)$ it is better to aim for $H_1^P$ leading to the trajectory in cyan. Since we deviate from the optimal 1v1 strategy, we also need to take the point $H_{-1}^P$ into account. However, due to the fact that $H_{-1}^P\notin \mathcal{B}_{r_e}(C_1^P)$ we cannot reach $H_{-1}^P$ in a turn-straight maneuver.
We thus use the strategy  $u_e^{*_{q_2}}(t)$,  $q_2\in \sign(x_{p_20}^E)$, in \eqref{eq:u_e*}, visualized in magenta as a candidate optimal solution.

\section{Non-static pursuers}
\label{sec:non-static-pursuers}

In this section we discuss extensions of Section \ref{sec:2v1} to the case of moving pursuers with $v_{p_1}=v_{p_2}\leq \frac{1}{2}v_e$.
Under this assumption and highlighted through Lemma \ref{lem:illustration_speed_ratio_assumption} we recall that the optimal pursuer strategies of the 1v1 setting are constant and the pursuer follows a straight line.
Thus, instead of focusing on particular pursuer strategies, we can focus on reachable sets for the pursuers and investigate the minimal time until the evader has a distance $\rho$ to the reachable sets. Moreover, due to the simple dynamics \eqref{eq:cartesianDynamicsPursuer}, the reachable set of a pursuer is defined through a ball.

\begin{corollary}
    For a set of pursuers $i\in \{1,2\}$ with dynamics \eqref{eq:cartesianDynamicsPursuer}, speed $v_{p_i}\in \R_{\geq 0}$ and initial positions $\xi_{p_i0}\in \R^2$, the points a pursuer can reach in time $t\in \R_{\geq 0}$ are characterized through $\mathcal{B}_{tv_{p_i}}(\xi_{p_i0})$, $i\in \{1,2\}$. 
\end{corollary}

In this case, we can conclude that the 1v1 game ends when $\xi_{e}(t) \notin \mathcal{B}_{t v_{p_i}+\rho}(\xi_{p_i0})$ is satisfied. Additionally, since the optimal strategy of the evader is characterized through the initial condition of the pursuer (but not the radius $\rho$), this extension does not change the strategy of the evader.

For the 2v1 setting, we have to extend the definition of the points $H_j^P$, $j\in \{-1,1\}$, introduced in Fig.~\ref{fig:2v1_setting}, which now depend on the initial position of the evader.
In particular, we denote the time the evader takes to reach the (unknown) point $H_j^P$ through $T^\# \in \R_{\geq 0}$. By using the input \eqref{eq:u_eH} and the solution \eqref{eq:opt_H_j_e}, $T^\#$ can be decomposed into 
\begin{align}
    T^{\#} = T_s +T_t,
\end{align}
where $T_t,T_s \in \R_{\geq 0}$ denote the time the evader follows a circle and goes straight, respectively.

With these definitions, considering the visualization in Fig.~\ref{fig:2v1_setting}, it holds that $H_j^P(1)=0$ and
\begin{align}
    \sqrt{H_j^P(2)^2 + (y_{p_1}^P)^2} = v_p T^\# + \rho. \label{eq:H_jP_condition}
\end{align}
Moreover, $v_p T^\#$ represents the distance covered by the pursuer and  $v_p T^\# +\rho$ implies that the pursuer is at a distance of $\rho$ to the evader when $\xi_{e}^P(T^\#)=H_j^P$. 
Here, $T_s$ and $T_t$ satisfy
\begin{align}
    T_t &= \tfrac{\bar{\beta}_q}{v_e}r_e, \qquad
    T_s = \tfrac{|\bar{S}_q^E -H_j^E|}{v_e},
\end{align}
which, used in \eqref{eq:H_jP_condition} together with the fact $|\bar{S}_q^P -H_j^P|=|\bar{S}_q^E -H_j^E|$, leads to $H_j^P(1) = 0$ and
\begin{align}
\begin{split}
    H_j^P(2)^2 + (y_{p_1}^P)^2 &= \left( \! \frac{v_p}{v_e} (\bar{\beta}_q r_e+|\bar{S}_q^P -H_j^P|) + \rho \! \right)^2 \! \! .
\end{split}\label{eq:H_jP_condition2}
\end{align}
Thus, for $j\in\{-1,1\}$, $H_j^P$ is defined through the solution of the equations \eqref{eq:H_jP_condition2},
\eqref{eq:alpha_bar}, 
\eqref{eq:def_S_bar},
\eqref{eq:beta_bar}
and the coordinate transformation \eqref{eq:co_trans_P_to_E}.
Note that, different to \eqref{eq:HP},
$H_{-1}^P(2)\neq - H_1^P(2)$, in general, when $v_{p}\neq 0$.
However, also note that the points $H_{j}^P$ do not depend on the pursuer selection in \eqref{eq:H_jP_condition} since $\xi_{p_1}^P(2)^2=\xi_{p_2}^P(2)^2$.

To take $H_{-1}^P(2)\neq - H_1^P(2)$ into account in the following result extending Lemma
\ref{lem:strategy_dominant}, we
need to introduce the cones 
\begin{align*}
\mathcal{C}_i^P = \{\xi_{p_i0}^P+a(H_{-1}^P-\xi_{p_i0}^P)+ b(H_{1}^P-\xi_{p_s0}^P) | \ a,b\in \R_{\geq 0} \} 
\end{align*}
for $i,s\in \{1,2\}$, $s\neq i$, extending the definitions \eqref{eq:terminatiing_sets}.
In particular, since $H_{-1}^P=-H_{1}^P$ in the case of static pursuers, it holds that $\bar{\mathcal{P}}_i^P(\rho) \subset \mathcal{C}_i^P$ for all $\rho>0$ and for all $i\in \{1,2\}$.

\begin{lemma} \label{lem:strategy_dominant_non_static}
    Consider Problem \ref{prob:game_of_degree} with $v_{p_1}=v_{p_2}\leq \frac{1}{2} v_e$.
    Let $\xi_{e0}^P\in \mathcal{B}_{\rho}(\xi_{p_10}^P)\cup \mathcal{B}_{\rho}(\xi_{p_20}^P)$,
    define 
$q_i\in -\sign(x_{p_i0}^E)$ for $i\in \{1,2\}$ and let 
$u_e^{*_{q_i}}(t)$ denote the optimal control strategies from the 1v1 setting  defined in
\eqref{eq:u_e*} with corresponding $T^i = V(\xi_{e\theta 0}^P,\{\xi_{{p_i}0}^P\}_{i=1}^2)$.
Let $H_j^P$, $j\in \{-1,1\}$, be uniquely defined through the solution of \eqref{eq:co_trans_P_to_E}, \eqref{eq:H_jP_condition2},
\eqref{eq:alpha_bar}--
\eqref{eq:beta_bar}.

If there exists $i\in \{1,2\}$ such that $\xi_{e}^P(T^i)\notin \mathcal{C}_i^P$ and $T^i>0$, then $u_e^{*_{q_i}}(t)$ is optimal for the 2v1 setting.
\end{lemma}

Before we give a proof, note that well-posedness of $H_j^P$ is similar to the condition $H_j^E\notin \mathcal{B}_{r_e}(C_q^E)$ in Lemma \ref{lem:optimal_strategy_to_point}, guaranteeing that the evader can reach the point through a turn-straight maneuver.

\begin{proof}
The proof follows the same arguments as the proof of  Lemma \ref{lem:strategy_dominant} with the only difference that the sets defined in \eqref{eq:terminatiing_sets} need to be replaced with $\mathcal{C}_i^P$. Here, $\xi_{e}^P(T^i)\notin \mathcal{C}_i^P$ ensures that the game terminates at time $T^i$, which is a lower bound for the 2v1 setting and is thus optimal.
\end{proof}

Similar to Lemma \ref{lem:strategy_dominant}, there is a direct extension of Lemma \ref{lem:strategy_same_sides} in the non-static pursuer case.

\begin{lemma} \label{lem:strategy_same_sides_non_static}
Consider Problem \ref{prob:game_of_degree} with $v_{p_1}=v_{p_2}\leq \frac{1}{2}v_e$.
    Let $\xi_{e0}^P\in \mathcal{B}_{\rho}(\xi_{p_10}^P)\cup \mathcal{B}_{\rho}(\xi_{p_20}^P)$, define 
$q_i\in -\sign(x_{p_i0}^E)$ for $i\in \{1,2\}$ and let 
$u_e^{*_{q_i}}(t)$ denote the optimal control strategies from the 1v1 setting  defined in
\eqref{eq:u_e*} with corresponding $T^i = V(\xi_{e\theta 0}^P,\{\xi_{{p_i}0}^P\}_{i=1}^2)$.
Let $H_j^P$, $j\in \{-1,1\}$, be uniquely defined through the solution of \eqref{eq:co_trans_P_to_E}, \eqref{eq:H_jP_condition2},
\eqref{eq:alpha_bar}--
\eqref{eq:beta_bar}.
Assume that $\xi_{e}^P(T^i)\notin \mathcal{C}_i^P$ for all $i\in \{1,2\}$.
For $j\in \{-1,1\}$, define $q_j\in \sign(H_{j}^E(1))$. Assume that  $H_{q_j}^E\notin \mathcal{B}_{r_e}(C_{q_j}^E)$ for all $j\in \{-1,1\}$.
For $j\in \{-1,1\}$ consider 
$u_e^{H_j^E}(t)$ and $\xi_e(t,u_e^{H_j^E})$ in 
\eqref{eq:u_eH}, \eqref{eq:opt_H_j_e}, respectively.  Let $\xi_e(T^{H_j},u_e^{H_j^E})=H_j^E$ and $k=\argmin_{j\in \{-1,1\}} T^{H_j}.$
Then $T^{H_k}$ is the smallest time such that $|\xi_{e0}-\xi_{p_i0}|\geq v_p T^{H_k} + \rho$ is satisfied.
\end{lemma}

\begin{proof}
Follows the same arguments as the proof of Lemma \ref{lem:strategy_same_sides}, but with arguments with respect to \eqref{eq:terminatiing_sets} replaced with arguments with respect to $\mathcal{C}_i^P$.
\end{proof}

\section{Conclusion, Limitations, and Future Work} \label{sec:conclusions}

We have discussed first results extending the 1v1 prying pedestrian differential game of \cite{braun_prying_2025} to a 2v1 setting. 
Some aspects of the extension seem straightforward, however, since the 2v1 setting leads to (at least) a four dimensional state space, a direct extension of the approaches in \cite{braun_prying_2025} does not seem to be tractable. 
We have thus used a different approach to leverage the results in \cite{braun_prying_2025} to derive results for the 2v1 game under the assumption that $v_{p_1}=v_{p_2}=0$ and $v_{p_1}=v_{p_2}\leq \frac{1}{2} v_{e}$, respectively. 
Thus, this paper is a stepping stone towards a complete solution of the 2v1 game. 
A limitation of the discussion and the formulation of Problem \ref{prob:game_of_degree}, which we have not touched in this paper is that the evader might be forced to reenter the game set as visualized in Fig.~\ref{fig:limitations} if the maximal turning rate $\omega_e$ is not sufficiently large.
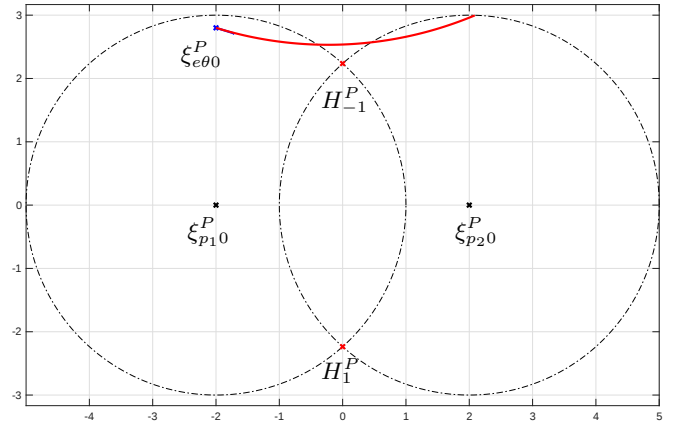
\begin{figure}[htb]
    \centering
    \begin{overpic}[width = 1\columnwidth]{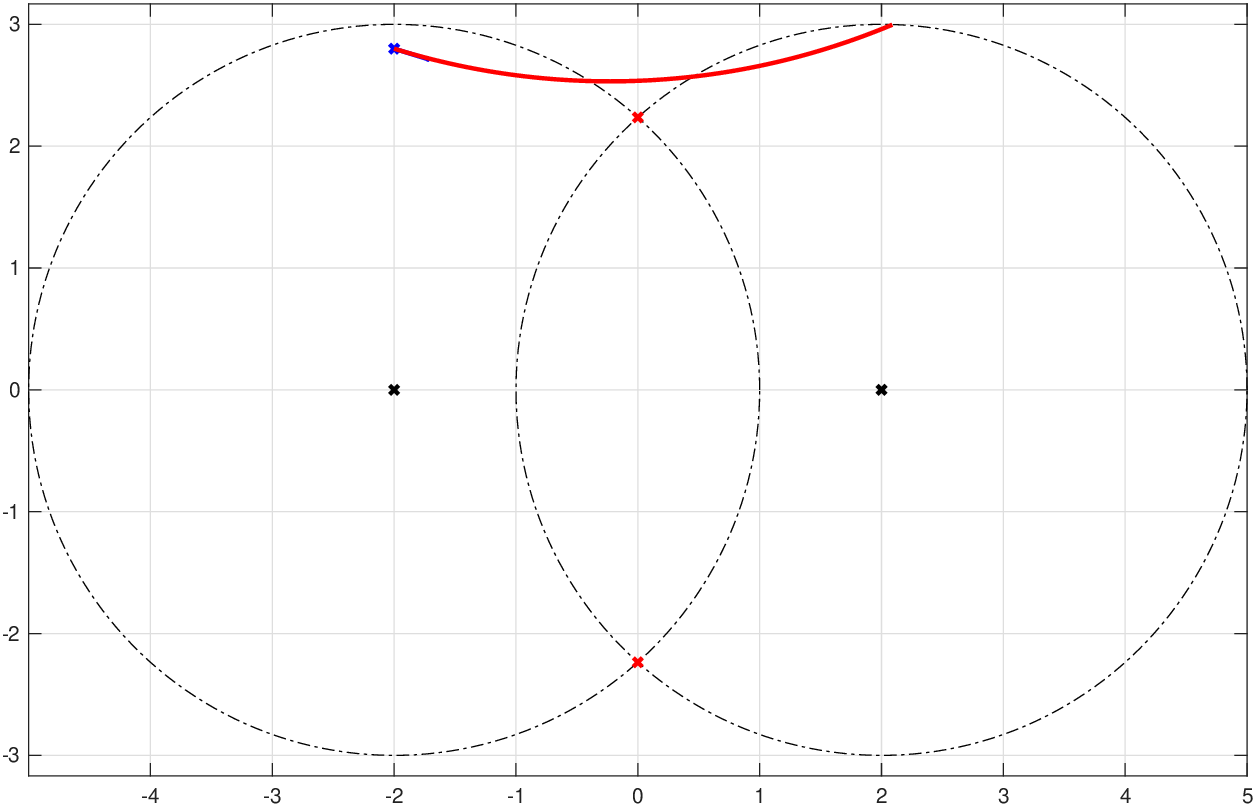}
    \put(26,55){\small{$\xi_{e\theta 0}^P$}}
    \put(47.5,6.5){\small{$H_1^P$}}
    \put(47.5,48){\small{$H_{-1}^P$}}
    \put(27,28){\small{$\xi_{p_10}^P$}}
    \put(68,28){\small{$\xi_{p_20}^P$}}
    \end{overpic}
    \caption{Illustration of limitations of Problem \ref{prob:game_of_degree}. According to Problem \ref{prob:game_of_degree}, the game ends when the set $\mathcal{B}_{\rho}(\xi_{p_10}^P)$ is left. However, depending on the turning rate $\omega_e$, the evader might be forced to reenter the surveillance radius of another pursuer.}.
    \label{fig:limitations}
\end{figure}
Thus, the current formulation of the 2v1 prying pedestrian game might not be complete. These issues, in  combination with a full solution of Problem \ref{prob:game_of_degree}, will be the focus of future work.

\bibliographystyle{IEEEtran}  
\bibliography{IEEEabrv,references.bib}

\begin{thebibliography}{10}
\providecommand{\url}[1]{#1}
\csname url@samestyle\endcsname
\providecommand{\newblock}{\relax}
\providecommand{\bibinfo}[2]{#2}
\providecommand{\BIBentrySTDinterwordspacing}{\spaceskip=0pt\relax}
\providecommand{\BIBentryALTinterwordstretchfactor}{4}
\providecommand{\BIBentryALTinterwordspacing}{\spaceskip=\fontdimen2\font plus
\BIBentryALTinterwordstretchfactor\fontdimen3\font minus
  \fontdimen4\font\relax}
\providecommand{\BIBforeignlanguage}[2]{{%
\expandafter\ifx\csname l@#1\endcsname\relax
\typeout{** WARNING: IEEEtran.bst: No hyphenation pattern has been}%
\typeout{** loaded for the language `#1'. Using the pattern for}%
\typeout{** the default language instead.}%
\else
\language=\csname l@#1\endcsname
\fi
#2}}
\providecommand{\BIBdecl}{\relax}
\BIBdecl

\bibitem{braun_prying_2025}
P.~Braun, T.~L. Molloy, and I.~Shames, ``Prying {Pedestrian}
  {Surveillance}-{Evasion}: {Minimum}-{Time} {Evasion} from an {Agile}
  {Pursuer},'' \emph{Journal of Guidance, Control, and Dynamics}, vol.~48,
  no.~9, pp. 2090--2104, Sep. 2025.

\bibitem{merz_homicidal_1974}
A.~W. Merz, ``The {Homicidal} {Chauffeur},'' \emph{AIAA Journal}, vol.~12,
  no.~3, pp. 259--260, Mar. 1974.

\bibitem{exarchos_asymmetric_2014}
I.~Exarchos and P.~Tsiotras, ``An asymmetric version of the two car
  pursuit-evasion game,'' in \emph{53rd {IEEE} {Conference} on {Decision} and
  {Control}}, Dec. 2014, pp. 4272--4277.

\bibitem{exarchos_suicidal_2015}
I.~Exarchos, P.~Tsiotras, and M.~Pachter, ``On the suicidal pedestrian
  differential game,'' \emph{Dynamic Games and Applications}, vol.~5, no.~3,
  pp. 297--317, 2015.

\bibitem{weintraub_introduction_2020}
I.~E. Weintraub, M.~Pachter, and E.~Garcia, ``An introduction to
  pursuit-evasion differential games,'' in \emph{2020 {American} {Control}
  {Conference} ({ACC})}.\hskip 1em plus 0.5em minus 0.4em\relax IEEE, 2020, pp.
  1049--1066.

\bibitem{dorothy_one_2024}
M.~Dorothy, D.~Maity, D.~Shishika, and A.~Von~Moll, ``One {Apollonius} {Circle}
  is enough for many pursuit-evasion games,'' \emph{Automatica}, vol. 163, p.
  111587, May 2024.

\bibitem{lewin_surveillance-evasion_1975}
J.~Lewin and J.~Breakwell, ``The surveillance-evasion game of degree,''
  \emph{Journal of Optimization Theory and Applications}, vol.~16, no.~3, pp.
  339--353, 1975.

\bibitem{lewin_conic_1979}
J.~Lewin and G.~Olsder, ``Conic surveillance evasion,'' \emph{Journal of
  Optimization Theory and Applications}, vol.~27, pp. 107--125, 1979.

\bibitem{lewin_isotropic_1989}
J.~Lewin and G.~J. Olsder, ``The isotropic rocket – {A} surveillance evasion
  game,'' \emph{Computers \& Mathematics with Applications}, vol.~18, no.~1,
  pp. 15--34, 1989.

\bibitem{merz_optimal_1973}
A.~Merz, ``Optimal evasive maneuvers in maritime collision avoidance,''
  \emph{Navigation}, vol.~20, no.~2, pp. 144--152, 1973.

\bibitem{molloy_optimal_2020}
T.~L. Molloy, T.~Perez, and B.~P. Williams, ``Optimal bearing-only-information
  strategy for unmanned aircraft collision avoidance,'' \emph{Journal of
  Guidance, Control, and Dynamics}, vol.~43, no.~10, pp. 1822--1836, 2020.

\bibitem{isaacs_differential_1965}
R.~Isaacs, \emph{Differential {Games}: {Mathematical} {Theory} with
  {Application} to {Warfare} and {Pursuit} {Control} and {Optimisation}}.\hskip
  1em plus 0.5em minus 0.4em\relax New York: Dover Publications, 1965.

\bibitem{ibragimov_simple_2018}
G.~Ibragimov, M.~Ferrara, A.~Kuchkarov, and B.~A. Pansera, ``Simple {Motion}
  {Evasion} {Differential} {Game} of {Many} {Pursuers} and {Evaders} with
  {Integral} {Constraints},'' \emph{Dynamic Games and Applications}, vol.~8,
  no.~2, pp. 352--378, Jun. 2018.

\bibitem{azamov_evasion_2025}
A.~Azamov, G.~Ibragimov, and A.~Kurbanov, ``Evasion {Differential} {Game} of
  {One} {Faster} {Evader} with {Bounded} {Maneuverability} from {Multiple}
  {Pursuers},'' \emph{Dynamic Games and Applications}, Oct. 2025.

\bibitem{garcia_multiple_2021}
E.~Garcia, D.~W. Casbeer, A.~Von~Moll, and M.~Pachter, ``Multiple {Pursuer}
  {Multiple} {Evader} {Differential} {Games},'' \emph{IEEE Transactions on
  Automatic Control}, vol.~66, no.~5, pp. 2345--2350, May 2021.

\bibitem{exarchos_uav_2016}
I.~Exarchos, P.~Tsiotras, and M.~Pachter, ``{UAV} collision avoidance based on
  the solution of the suicidal pedestrian differential game,'' in \emph{{AIAA}
  {Guidance}, {Navigation}, and {Control} {Conference}}, 2016, p. 2100.

\bibitem{hagedorn_differential_1976}
P.~Hagedorn and J.~V. Breakwell, ``A differential game with two pursuers and
  one evader,'' \emph{Journal of Optimization Theory and Applications},
  vol.~18, no.~1, pp. 15--29, Jan. 1976.

\bibitem{molloy_minimum-time_2023}
T.~L. Molloy and I.~Shames, ``Minimum-{Time} {Escape} from a {Circular}
  {Region} for a {Dubins} {Car},'' \emph{IFAC-PapersOnLine}, vol.~56, no.~1,
  pp. 43--48, 2023.

\bibitem{weintraub_minimum_2024}
I.~E. Weintraub, A.~Von~Moll, and M.~N. Pachter, ``Minimum {Time} {Escape} from
  a {Circular} {Region} of a {Dubins} {Car},'' in \emph{{NAECON} 2024 - {IEEE}
  {National} {Aerospace} and {Electronics} {Conference}}, 2024, pp. 34--37.

\bibitem{braun_capture_2023}
P.~Braun, I.~Shames, D.~Hubczenko, A.~Dostovalova, and B.~Fraser, ``Capture the
  flag games: {Observations} from the 2022 {Aquaticus} competition,''
  \emph{IFAC-PapersOnLine}, vol.~56, no.~2, pp. 11\,363--11\,368, 2023.

\bibitem{dubins_curves_1957}
L.~E. Dubins, ``On curves of minimal length with a constraint on average
  curvature, and with prescribed initial and terminal positions and tangents,''
  \emph{American Journal of mathematics}, vol.~79, no.~3, pp. 497--516, 1957.

\end{thebibliography}

\end{document}